\documentclass[lettersize,journal]{IEEEtran}
\usepackage{amsmath,amsfonts}
\usepackage{algorithmic}
\usepackage{algorithm}
\usepackage{array}
\usepackage{textcomp}
\usepackage{stfloats}
\usepackage{url}
\usepackage{verbatim}
\usepackage{graphicx}
\usepackage{cite}
\usepackage{hyperref}
\usepackage{amsthm}
\usepackage{caption}         
\usepackage{subcaption}      
\usepackage{booktabs} 
\usepackage{multirow} 
\usepackage{tcolorbox}               
\usepackage{enumitem}                
\newtheorem{theorem}{Theorem}

\usepackage{balance}
\begin{document}
\title{EIFL: Efficiently Protecting Global Model Privacy and Integrity Against an Untrusted Server in Federated Learning}
\author{Zehui Liao, Qiang Li and Binghui Wang,~\IEEEmembership{Senior Member,~IEEE}
\thanks{Zehui Liao and Qiang Li are with the College of Computer Science and Technology, Jilin University, Changchun, China. (E-mail: liaozh25@mails.jlu.edu.cn)}
\thanks{Binghui Wang is with the Department of Computer Science, Illinois Institute of Technology, Chicago, IL, USA. Wang is the corresponding author (bwang70@iit.edu).}}


\maketitle

\begin{abstract}
Federated learning (FL) typically adopts a server-client architecture, where the server aggregates clients' local models (i.e., the input) and returns the aggregated global model (i.e., the output) to clients. An untrusted server may return a tampered global model to compromise the output integrity. Some existing schemes focus on verifying the output integrity. However, these schemes mostly rely on clients pre-negotiating a set of identical auxiliary information among themselves and keeping it confidential from the server. This dependency creates a verification vulnerability: if the auxiliary information is leaked, the verification method may be circumvented. Additionally, in some privacy-sensitive scenarios (e.g., commercial federated learning), the global model may need to be kept confidential from an untrusted server. However, only a few works achieve the output privacy while addressing verification vulnerability, at the cost of prohibitive computation and communication overhead. To address these challenges simultaneously, we propose a novel \emph{provably} privacy-preserving FL method called EIFL. Specifically, we adopt a two-stage aggregation and combine it with symmetric encryption to protect the output privacy. To address the verification vulnerability, we propose an efficient verification method based on vector inner product for output integrity, and a random vector generation method for clients to agree on auxiliary information. EIFL innovatively binds the auxiliary information to the output integrity, and eliminates the need to keep the auxiliary information confidential from the server. Moreover, EIFL is robust against client dropout during the verification phase through a simple resending operation. Evaluation results validate the advantages of EIFL over state-of-the-art schemes in terms of computation and communication overhead.
\end{abstract}


\section{Introduction}

Federated learning (FL) \cite{DBLP:cite1} is an emerging collaborative learning framework that aims to address the data sharing problem.
In FL, the participating clients train their models locally, and only share the trained local models, instead of the raw data with a center server.   
The server updates the global model by aggregating received client models, and broadcasts the updated global model to all participating clients.

Though FL was originally designed to protect clients' data privacy, existing works \cite{DBLP:cite2,DBLP:cite4,DBLP:cite40} have shown that the shared local models can also leak clients' sensitive information. Secure aggregation \cite{DBLP:cite5} addresses this issue, which enables the server to obtain only the aggregated global model without revealing clients’ local models. However, when an untrusted server undertakes the aggregation task, this raises additional concerns for clients: (1) clients may be concerned that the untrusted server will return a tampered global model \cite{DBLP:cite7,DBLP:cite8,VCD-FL,vsersa}, and (2) in some privacy-sensitive scenarios such as commercial FL, clients might be unwilling to expose the commercially used global model to an untrusted server \cite{PPDL,PEVSA}.

Some schemes \cite{DBLP:cite8,SVFL,PVFL,DBLP:cite7,VOSA,vsersa} focus on verifying the {\bf output integrity} (i.e., preventing the global model from being tampered with by the untrusted server). However, these schemes rely on clients pre-negotiating identical auxiliary information among themselves and keeping it confidential from the server \cite{PEVSA}. For example, VERSA \cite{vsersa} requires all clients to hold the same secret vector and keep it confidential from the server. Once the auxiliary information is leaked, the malicious server may forge an aggregated output and bypass the output integrity verification. To address this potential vulnerability, VCD-FL \cite{VCD-FL} moves the negotiation of auxiliary information to the phase after the server’s computation. In essence, this scheme still relies on the confidentiality of auxiliary information, because the auxiliary information has not yet been generated when the server aggregates the global model. At the cost, the scheme also requires additional peer-to-peer communication among clients to negotiate auxiliary information, which incurs substantial communication overhead. Some other schemes aim to encrypt aggregated results based on secure aggregation to achieve {\bf output privacy} (i.e., preventing the untrusted server from accessing the global model). Classical schemes \cite{PEVSA} use homomorphic encryption to achieve aggregation over ciphertexts. Recently, some schemes \cite{VCD-FL,VFL,SWAS} have introduced Lagrange interpolation to achieve more efficient output privacy protection. However, the aforementioned schemes incur substantial computation overhead.

We focus on the output privacy and verification vulnerability of output integrity in this paper. 
We address the above challenges by proposing a novel \emph{provably}  privacy-preserving FL method, called EIFL. Specifically, EIFL adopts a two-stage aggregation and combines it with symmetric encryption to protect the output privacy, where the server only aggregates an intermediate result, while the computation to obtain the global model is performed locally by clients.
To address the verification vulnerability of output integrity, EIFL introduces an efficient verification method based on vector inner product and a random vector generation method. 
Specifically, clients in EIFL use the hash value of the intermediate result from the server as the seed of a pseudorandom generator ({\bf PRG}) to agree on a random vector. Clients then compute the inner product between the random vector and their local model vectors and exchange these inner product values via the server. After locally aggregating to obtain the global model, each client calculates the inner product between the global model vector and the random vector, and verifies the output integrity by comparing this value with the sum of all inner products uploaded by other clients.
Benefiting from the collision resistance of SHA-256 hash function and the randomness of the ${\bf PRG}$, our verification method can detect the server's malicious behavior with an extremely high probability when the server returns an incorrect intermediate result. To the best of our knowledge, EIFL makes the first attempt to disclose the auxiliary information, and innovatively binds the auxiliary information to the output integrity. Compared with existing methods, EIFL does not require additional communication rounds, nor does it require hiding the auxiliary information from the server.
Furthermore, EIFL can also handle dropout clients during the verification phase through a simple resending operation.

We emphasize that all these security goals are theoretically guaranteed.  
Evaluation results validate the effectiveness and efficiency of EIFL and its advantages over the state-of-the-art schemes in terms of computation and communication overheads. Our main contributions are summarized as follows:
\begin{itemize}
\item We adopt a two-stage aggregation and combine it with symmetric encryption to protect the output privacy, which has extremely low computation overhead.
\item We propose an efficient verification method for output integrity based on vector inner product, the communication overhead of which is independent of the dimension of the model vector.
\item We design an efficient random vector generation method for clients to agree on a random vector. This method makes the first attempt to disclose the auxiliary information, and innovatively binds the auxiliary information to the output integrity to address the verification vulnerability. 
\item We introduce a simple resending operation to render EIFL robust against client dropout during the verification phase.
\end{itemize}

The remainder of this paper is structured as follows. Section \ref{sec1} reviews some related work. Section \ref{pre} introduces the related preliminaries. Next, Section \ref{Design} details the various stages of EIFL. Afterwards, Sections \ref{Analysis} and \ref{Evaluation} detail the comprehensive theoretical analysis and evaluation, respectively. Finally, Section \ref{conclusion} concludes our work.

\section{Related Work}
\label{sec1}
Federated learning (FL) is a widely adopted privacy-preserving distributed learning framework, which enables collaborative model training across edge clients without exposing their raw private data. However, practical FL deployment still faces critical challenges in privacy preservation of local updates, robustness against Byzantine attacks from malicious clients, and security guarantees against untrusted aggregation servers. Existing studies have addressed these challenges along three main technical lines: foundational secure aggregation schemes, Byzantine-robust aggregation against adversarial clients, and enhanced secure aggregation for untrusted server scenarios. We systematically review representative works in these three directions below, and analyze their limitations to motivate our proposed EIFL scheme.

\subsection{Secure Aggregation}
Secure aggregation in federated learning (FL) refers to the aggregate computation of the sum of local models in a secure way without learning any information about the local models. 
Existing schemes can be mainly divided into three categories: differential privacy (DP)~\cite{pathak2010multiparty,DBLP:cite23,hamm2016learning,mcmahan2018learning,wei2020federated}, secure multi-party computation (MPC)~\cite{SecureML,DBLP:cite5,DBLP:cite39}, and homomorphic encryption (HE)~\cite{PPDL,DBLP:cite14,PEVSA}. For example, \cite{wei2020federated} proposed a novel DP FL framework, in which carefully designed noises are added to the local client models before aggregation. However, the current DP-based schemes have high utility losses. MPC-based schemes ensure local clients and the server to jointly complete the aggregation without disclosing the clients' private data. The first MPC-based secure aggregation proposed a double-mask protocol \cite{DBLP:cite5} to protect clients' local model vectors. Subsequent work \cite{DBLP:cite39} presented a more communication-efficient version based on a k-neighbors network. However, the excessive number of communication rounds required for protocol execution is a fundamental inherent limitation of the MPC-based secure aggregation schemes. The work \cite{PPDL} used additively homomorphic encryption in asynchronous stochastic gradient descent training for a neural network. Since the server performs aggregation operations directly on the encrypted ciphertexts without decryption, it cannot access any sensitive information associated with the participating clients' local models, thus guaranteeing the privacy of client-side data. However, the inherently prohibitive computation overhead of homomorphic encryption remains a critical fundamental challenge for its large-scale application in federated learning.

\subsection{Secure Aggregation Against Untrusted Server} 
Traditional secure aggregation schemes allow the server to aggregate the global model without revealing clients' local models. Considering potential malicious behaviors of untrusted servers and some privacy-sensitive scenarios, some variant schemes have been proposed to protect output privacy and integrity.

For instance, some schemes aim to encrypt aggregated results to protect them from being accessed by the untrusted server. Classical schemes \cite{PPDL,DBLP:cite14,PEVSA} use homomorphic encryption to achieve aggregation over ciphertexts. Recently, some schemes \cite{VCD-FL,VFL,SWAS} have introduced Lagrange interpolation to achieve more efficient output privacy protection. However, the aforementioned schemes incur substantial computation overhead. 

Some schemes focus on verifying aggregated results to prevent untrusted servers from compromising output integrity. Mainstream schemes use homomorphic hash \cite{DBLP:cite8,SVFL,PVFL} or bilinear pairing \cite{DBLP:cite7,VOSA} to achieve verifiability. Some schemes \cite{vsersa,likeversa} introduce random transformations to achieve more efficient verification. However, the aforementioned schemes rely on clients pre-negotiating identical and confidential auxiliary information, which posed a serious security threat. VCD-FL \cite{VCD-FL} addresses this issue by moving the negotiation of auxiliary information to the phase after the server's computations. However, its peer-to-peer negotiation method incurs substantial communication overhead. 

\subsection{Byzantine-robust Aggregation Against Malicious Clients}
Byzantine-robust aggregation in FL is an inherent challenge, in which malicious clients can send poisoned updates to the aggregation server to influence the global model. Existing schemes \cite{DBLP:cite33,DBLP:cite34,DBLP:cite35,DBLP:cite36,Multi-metrics,MESAS,FLAME,Lockdown,rlr,FLTrust} typically employ statistical analysis to filter out or clip malicious model updates before aggregating them, and assume the existence of a reliable aggregation server. Krum \cite{DBLP:cite36} is the first solution to defend against Byzantine attacks. Specifically, Krum selects multiple reliable local model updates by identifying those with the smallest sum of squared Euclidean distances to others. \cite{DBLP:cite35} proposed a median-based aggregation, where the server obtains the $j$-th parameter of its global model by calculating the median of the $j$-th parameter of all clients' local models. \cite{Bulyan} proposed Bulyan, which combines the idea of Krum and median. Bulyan first iteratively uses Krum to select a portion of client's local models. Then, Bulyan aggregates these clients' local models using a variant of the trimmed mean. These methods can defend against malicious clients to some extent, but are still not effective enough. To further enhance the performance, FLTrust \cite{FLTrust} assumes that the server has a clean and small dataset to bootstrap trust. FLAME \cite{FLAME} combines differential privacy (DP), model clustering, and weight clipping to eliminate backdoors.

\section{Preliminaries}
\label{pre}
In this section, we first explain some preliminaries and cryptographic primitives used in our approach, which should facilitate understanding the proposed protocol. Then we define the threat model and introduce the design goals of our EIFL.
\subsection{Federated Learning.} Suppose there are $N$ clients $C=\{C_n\}_{n=1}^N$ in FL, and each client $C_n$ has a private dataset ${D}_n$. At the beginning of each round $j$, the server sends a global model ${\bf g}^{j}$ to all clients. Each client $C_n$ uses the global model as an initialization and updates its local model ${\bf g}_n^j$  by locally minimizing the empirical loss over its dataset ${D}_n$. 
The updated local models of all clients are then uploaded to the server for global model aggregation. 
For instance, in FedAvg~\cite{DBLP:cite1}, the server takes a weighted average of these models and updates the global model for the next round $j+1$ as ${\bf g}^{j+1} \leftarrow \lambda\sum_{n=1}^N {|{D}_n|}/{d} \cdot {\bf g}_n^{j}$, 
where $\lambda$ is the learning rate and $d=\sum_{n=1}^N |D_n|$. The training process continues until reaching convergence or a predefined number of rounds.

\subsection{Secret Sharing}

In EIFL, we use the Shamir’s $t$-out-of-$N$ secret sharing protocol \cite{DBLP:cite16} to deal with dropout clients  and preserve clients' {input and output} 
privacy, where $N$ is the number of clients in our EIFL, and $t$ is the minimum number of clients required to recover the secret. The secret sharing protocol involves the following two steps: 

\begin{enumerate}
\item {\bf SS.Share}$\left(s,t,C\right) \rightarrow \left\{\left(n,s_n\right)\right\}_{n\in C}$: This algorithm takes as input a secret $s$, the threshold value $t$ and the set of clients $C$, where $|C|=N \geq t$. The output of this algorithm is a set of secret shares, denoted by $\left\{s_n\right\}$, each of which is assigned to a distinct holder (client) $n \in C$.
\item {\bf SS.Recon}$\left(t,\left\{\left(n,s_n\right)\right\}_{n\in P}\right) \rightarrow s$: This algorithm takes as input a subset $P$ of shares, where $n\in P\subseteq C$ and $t\leq|P|$, and outputs a secret $s$.
\end{enumerate}

\subsection{Symmetric Encryption}
We use a symmetric encryption scheme to protect the clients’ privacy and authenticate each client’s identity. For security, we require {\bf SE} to be \emph{indistinguishable} under a chosen plaintext attack (IND-CPA) and ciphertext integrity (IND-CTXT) as defined in \cite{DBLP:cite20}. The symmetric encryption scheme includes the following algorithms:
\begin{enumerate}
\item {\bf SE.Gen}$\left(k\right) \rightarrow SK$. The key generation algorithm takes as input a security parameter $k$ and outputs a secret symmetric key $K$.

\item {\bf SE.Enc}$\left(m,SK\right) \rightarrow c$. This algorithm inputs a plaintext $m$ to be encrypted and a symmetric key $SK$, and outputs the ciphertext $c$ of $m$ under the key $SK$.

\item {\bf SE.Dec}$\left(c,SK\right) \rightarrow m$. This algorithm inputs a ciphertext $c$ to be decrypted and its corresponding symmetric key $SK$, and outputs the plaintext $m$.
\end{enumerate}

\subsection{Key Agreement}
We utilize Diffie-Hellman key agreement protocol \cite{DBLP:cite18} to create the shared key for any two clients, which involves the following steps: 
\begin{enumerate}
\item  {\bf KA.Param}$\left(k\right) \rightarrow \left(pp\right)$. The parameter generation algorithm samples the public parameter $pp$ {(including a group $G$ with prime order $q$ and its generator $g$) that has a security parameter $k$.} 

\item  {\bf KA.Gen}$\left(pp\right) \rightarrow \left(SK,PK\right)$. The key generation algorithm samples a public key $PK$ and a secret key $SK$ pair from the public parameters $pp$.

\item  {\bf KA.Agree}$\left(SK_n,PK_m\right) \rightarrow s_{m,n}$. This algorithm takes as input the secret key $SK_n$ of client $n$ and the public key $PK_m$ of client $m$ and outputs a private agreed key $s_{m,n}$ between client $n$ and $m$. 
\end{enumerate}

\subsection{Signature Scheme}
We use the one-time ECDSA digital signature scheme based on SHA-256 to ensure the integrity of each client's masked gradient vector and authenticate their identities. Here we regard the masked gradient vector of each client as a message and sign it. The signature scheme includes the following algorithms:

\begin{enumerate}
\item {\bf SIG.Gen}$\left(k\right) \rightarrow \left(SK,PK\right)$. The key generation algorithm takes as input a security parameter $k$ and outputs a secret key $SK$ and a public key $PK$.

\item {\bf SIG.Sign}$\left(SK,m\right) \rightarrow \sigma$. This algorithm inputs the secret key $SK$ and message $m$ and outputs a signature $\sigma$.

\item {\bf SIG.Ver}$\left(PK,m,\sigma \right) \rightarrow \left\{0,1\right\}$: This algorithm takes as input a public key $PK$, a message $m$ and a signature $\sigma$, and returns a bit indicating whether the signature should be considered valid or not.
\end{enumerate}

\subsection{Pseudorandom Generator}
We employ a secure pseudorandom generator $\bf PRG$ that takes a seed and produces a random number as an output. The output must be computationally indistinguishable from a uniform element sampled from the output space as long as the seed is hidden from the distinguisher. In addition, for the same seed, the output of the $\bf PRG$ is identical.

\subsection{Threat Model}
\label{threat model}

We assume that the clients are semi-honest, i.e., they strictly perform operations in accordance with the agreed-upon FL protocol, but try to infer other clients' \emph{input privacy} \cite{DBLP:cite2,DBLP:cite3,DBLP:cite4}. Since the double-mask secure aggregation protocol \cite{DBLP:cite5} has addressed the issue of the input privacy, we focus on the security and privacy vulnerability with respect to output privacy and output integrity posed by untrusted server.

We suppose that the server may be captured by a malicious adversary. Specifically, the server may try to infer the output privacy using known information. In terms of output integrity, the server may deviate the protocol or return an incorrect result to disrupt model availability, which is the same with existing schemes \cite{DBLP:cite7,DBLP:cite8,vsersa}. Furthermore, we suppose that the server may forge aggregated results and proofs using the public auxiliary information used for verification. Note that existing schemes rely on the public auxiliary information used for verification to be kept confidential from the server, while our scheme lifts this restriction.

\textit{Remark 1.} Clients are not allowed to intentionally upload incorrect information to assist the malicious server in disrupting or forging aggregation results, as integrity verification under such circumstances has been demonstrated to be infeasible in previous studies \cite{mul,vsersa}. How to achieve Byzantine-robust aggregation against malicious clients in FL is another inherent challenge \cite{DBLP:cite36,FLAME,FLTrust}.

\subsection{Design Goals}
Our protocol is designed to efficiently protect the privacy and integrity of aggregation results. Specifically, we describe our objectives as follows:
\begin{itemize}
\item {\bf Output privacy.} The server cannot learn any useful information about the aggregation results.
\item {\bf Output integrity.} The server can not forge an output and pass verification.
\item {\bf Robust verification.} Even if some clients drop out mid-execution, our protocol is still able to finalize the verification process without interruption.
\item {\bf Efficient computation and communication.} Compared with the original secure aggregation protocol, our protocol incurs only lightweight additional overhead to guarantee both output privacy and output integrity.
\end{itemize}
\textit{Remark 2.} We explicitly clarify that the privacy protection of clients' local models is not claimed as a novel contribution of this work. Consistent with mainstream verifiable FL frameworks \cite{DBLP:cite7,DBLP:cite8,vsersa,SVFL}, our proposed EIFL is designed to operate as an upper-layer protocol on top of double-mask secure aggregation protocol \cite{DBLP:cite5}, and fully inherits the privacy guarantees of the underlying protocol.
\section{EIFL Design}
\label{Design}

\subsection{Overview}

\begin{figure*}
  \centering
\includegraphics[scale=0.45]{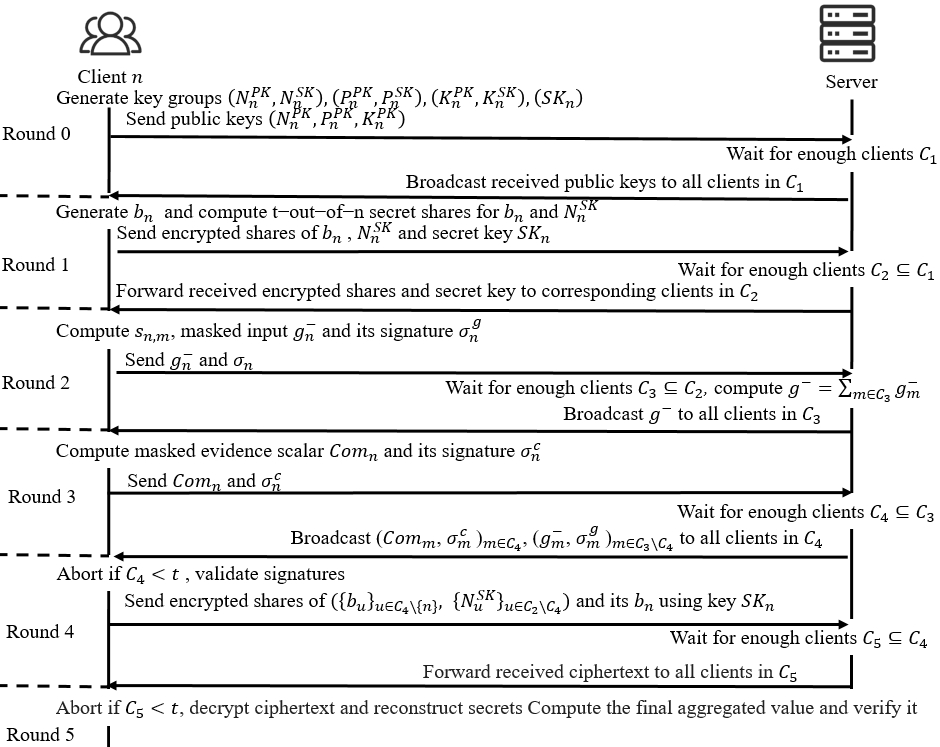}
  \caption{High-level overview of EIFL. We use a single client for simplicity.}
  \label{fig:fig12}
  \vspace{-3mm}
\end{figure*}

EIFL is designed to run on top of double-mask protocol \cite{DBLP:cite5}, aiming to protect output privacy and address the verification vulnerability of output integrity. EIFL goes through 6 rounds, among which Round 2 and Round 5 correspond to the two stages of aggregation respectively. Round 0 and Round 1 are mainly used for information exchange among clients. Clients in Round 2 mask their local models in accordance with the double-mask protocol and upload them to the server. Then the server directly aggregates the masked models into an intermediate result, which is the first aggregation stage. Clients in Round 3 use the proposed random vector generation method and the intermediate result to agree on a random vector. Then they compute the inner product between the random vector and their local model vectors and exchange the values via the server. Round 4 is used for clients to exchange encrypted information in preparation for the aggregation in the second stage. In Round 5, clients decrypt the encrypted messages received in Round 4 and aggregate the global model locally, which constitutes the second aggregation stage. In addition, they also need to use the vector inner product of the global model and the agreed random vector to verify the output integrity.
Figure \ref{fig:fig12} shows the high-level overview of EIFL and Table \ref{tbl:notations} shows the important notations used in this paper.

\begin{table}
\center
\caption{Important notations used in the paper.}
\addtolength{\tabcolsep}{-3pt}
\label{tbl:notations}
\begin{tabular}{c|c}
\hline
{\bf Notation} & {\bf Meaning}  \\  \hline
$k$& the security parameter \\ \hline
$t$& the threshold of secret sharing \\ \hline
$N$& the total number of clients $C_i$ \\ \hline
$M$& the model dimension \\\hline
$C_i$& the set of online clients in round $i$  \\ \hline
$SK_n$& the symmetric encryption secret key of client $n$  \\ \hline
$(N_n^{PK},N_n^{SK})$&the public/private keys for agreed masks \\ \hline
$\left(P_n^{PK},P_n^{SK}\right)$& the public/private keys for encrypting shares \\ \hline
$\left(K_n^{PK},K_n^{SK}\right)$& the signature public/private keys \\ \hline
${\bf g}_n$& the local model of client $n$ \\ \hline
${\bf\tilde g}_n$& the masked local model of client $n$ \\ \hline
${\bf g}$& the sum of local models for clients in $C_4$ \\ \hline
$\sigma_n^g$& the signature of ${\bf\tilde g}_n$\\ \hline
${\bf\tilde g}$& the intermediate result generated by the server \\ \hline
${\bf V}$& the random vector generated from ${\bf\tilde g}$\\ \hline
$b_n$& the random number generated by client $n$  \\ \hline
$s_{n,m}$& the agreed mask of client $n$ and client $m$ \\ \hline
$Com_n$ &the masked inner product value of client $n$  \\ \hline
$\sigma_n^c$ & the signature of $Com_n$  \\ \hline 
$Com$ &the sum of inner product values for clients in $C_4$  \\ \hline
\end{tabular}
\end{table}

\subsection{Workflow}
Our EIFL workflow consists of the following rounds. In the sequel, we use boldface letters to denote vectors, and $\cdot$ between two vectors denotes the inner product operation. Figure \ref{fig:fig3} presents the full details of EIFL.
\begin{itemize}
\item {\bf Round 0:} Each client $n\in C$ generates key $\left(SK_n\right)$,  $\left(N_n^{PK},N_n^{SK}\right)$, $\left(P_n^{PK},P_n^{SK}\right)$, $\left(K_n^{PK},K_n^{SK}\right)$  needed in EIFL and sends all public keys $\left(N_n^{PK},P_n^{PK},K_n^{PK}\right)$ to the server. After receiving messages from at least $t$ clients (denoted as $C_1 \subseteq C$), the server broadcasts received public keys to all clients in $C_1$.

\item {\bf Round 1:} Each client $n\in C_1$ generates a random number $b_n$ and computes the $t$-out-of-$N$ secret shares for $b_n$ and $N_n^{SK}$ as:
\begin{equation}
    \begin{aligned}
        &\left\{\left(m,N_{n,m}^{SK}\right)\right\}_{m\in C_1 }\leftarrow {\bf SS.Share}(N_n^{SK},t,C_1)\\
        &\left\{\left(m,b_{n,m}\right)\right\}_{m\in C_1} \leftarrow {\bf SS.Share}\left(b_n,t,C_1\right)
    \end{aligned}
\end{equation}
For each other client $m\in C_1\backslash\left\{n\right\}$, client $n$ first computes the agreed key $k_{n,m}\leftarrow {\bf KA.Agree}\left(P_n^{SK},P_m^{PK}\right)$ via the Diffie-Hellman key agreement protocol, then encrypts the shares and its secret key $SK_n$ using the agreed key $k_{n,m}$ as:
\begin{equation}
    \begin{aligned}
        E_{n,m}\leftarrow {\bf SE.Enc}\left(N_{n,m}^{SK} || b_{n,m} ||SK_n,k_{n,m}\right)
    \end{aligned}
    \nonumber
\end{equation}
The encryption step ensures that the server cannot obtain the plaintext secret shares and $SK_n$, which is a critical prerequisite for preserving the privacy of clients' local models and the aggregated result.
Finally, client $n$ needs to submit the ciphertext set $\left\{\left(n,m,E_{n,m}\right)\right\}_{m\in C_1\backslash\left\{n\right\}}$ to the server. After receiving messages from at least $t$ clients, the server forms the set of surviving clients $C_2\subseteq C_1$, and forwards received encrypted shares and secret keys to all clients in $C_2$.

\item {\bf Round 2:} Each client $n \in C_2$ decrypts ciphertexts and computes the agreed masks $s_{n,m}\leftarrow {\bf KA.Agree} \left(N_n^{SK},N_m^{PK}\right)$ with all the other clients $m$, where $s_{n,m}=s_{m,n}$. Then, following the double-mask protocol, the client $n$ masks its local model ${\bf g}_n$ to ${\bf\tilde g}_n$ using the agreed masks and its $b_n$ as

\begin{equation}
    \begin{aligned}
        {\bf\tilde g}_n = {\bf g}_n + {\rm\bf PRG}\left(b_n\right) + \sum_{m\in C_2:n<m}{\rm\bf PRG}\left(s_{n,m}\right)\\
     - \sum_{m\in C_2:n>m}{\rm\bf PRG}\left(s_{m,n}\right)
    \end{aligned}
    \label{local-mask}
\end{equation} 
As the server must retransmit the local model vectors of clients that drop out of the protocol in {\bf Round 3}, each client is required to generate a digital signature for its masked model vector before submitting it. This mechanism prevents the server from sending forged local model vectors to clients. Therefore, client $n$ signs ${\bf\tilde g}_n$ as $\sigma_n^g\leftarrow {\bf SIG.Sign}\left(K_n^{SK},{\bf\tilde g}_n\right)$ and submits $\left({\bf\tilde g}_n,\sigma_n^g\right)$ to the server. After receiving messages from at least $t$ clients (denoted as $C_3\subseteq C_2$), the server sums all the masked local models as ${\bf\tilde g} = \sum_{m\in C_3}{\bf\tilde g}_m$ and broadcasts ${\bf\tilde g}$ to all clients in $C_3$.

\textit{Remark 3.} Note that the intermediate result of $C_3$ is a masked global model, which means that each client still needs to locally remove the mask to obtain the global model.

\item {\bf Round 3:} Each client $n \in C_3$ computes random vector ${\rm\bf V} = {\rm\bf PRG}\left(Hash\left({\bf\tilde g}\right)\right)$. Since the vector ${\bf V}$ is randomly generated from the intermediate result ${\bf\tilde g}$ and used to compute the inner product with the local input vector ${\bf g}_n$ for verifying the output integrity, it is inherently bound to the output integrity. Benefiting from the collision resistance of SHA-256 hash function and the randomness of the ${\bf PRG}$, our verification method can detect the server's malicious behavior with an extremely high probability when the server returns an incorrect intermediate result. To prevent information leakage of the private local model ${\bf g}_n$ from the plaintext inner product ${\bf V}\cdot{\bf g}_n$, client $n$ masks the inner product value using the same double-mask mechanism as the model aggregation process:

\begin{equation}
\begin{aligned}
Com_n = {\rm\bf V}\cdot{\rm\bf g_n} +{PRG}\left(b_n\right) + \sum_{m\in C_3:n<m}{PRG}\left(s_{n,m}\right)\\ - \sum_{m\in C_3:n>m}{PRG}\left(s_{n,m}\right)
\end{aligned}
\label{com-mask}
\end{equation}
In addition, to prevent the server from sending forged inner product values to clients, client $n$ must generate the signature $\sigma_n^c$ of $Com_n$ as $\sigma_n^c\leftarrow {\bf SIG.Sign}\left(K_n^{SK},Com_n\right)$ and sends $\left(Com_n,\sigma_n^c\right)$ to the server. After receiving messages from at least $t$ clients (denoted as $C_4\subseteq C_3$), the server broadcasts messages $\left\{\left(m,Com_m,\sigma_m^c\right)\right\}_{m\in C_4}$ and masked input vectors of dropout clients in Round 3 (i.e., $\left\{\left(m,{\bf\tilde g}_m,\sigma_m^g\right)\right\}_{m\in C_3\backslash C_4}$) to all clients in $C_4$. Since only clients in $C_4$ upload inner product values while the intermediate result includes the masked input vectors of dropout clients, the server must resend the masked input vectors of dropout clients for other clients to remove them from the intermediate result. Specifically, the server broadcasts $\left(\left\{\left(m,Com_m,\sigma_m^c\right)\right\}_{m\in C_4},\left\{\left(m,{\bf\tilde g}_m,\sigma_m^g\right)\right\}_{m\in C_3\backslash C_4}\right)$ to all clients in $C_4$.

\item {\bf Round 4:} Each client $n \in C_4$ decides to accept or reject the received masked local models by verifying the corresponding signatures. If all signatures are correct, client $n$ first computes the agreed key $k_{n,m}\leftarrow {\bf KA.Agree}\left(P_n^{SK},P_m^{PK}\right)$ for each other client $m\in C_2\backslash \left\{n\right\}$. Using the agreed key $k_{n,m}$, client $n$ can decrypt the ciphertext as:
$$N_{m,n}^{SK} ||b_{m,n} ||SK_m\leftarrow {\bf SE.Dec}\left(E_{m,n}, k_{n,m}\right)$$ 
Client $n$ packages its $b_n$, the shares $\left\{b_{m,n}\right\}_{m\in C_4\backslash\left\{n\right\}}$ and $\left\{N_{m,n}^{SK}\right\}_{m\in C_2\backslash C_4}$ as: 
\begin{equation}
\begin{aligned}
H_n=\left(b_n, \left\{b_{m,n}\right\}_{m \in C_4 \backslash\left\{n\right\}},\left\{N_{m,n}^{SK}\right\}_{m \in C_2 \backslash C_4}\right)
\end{aligned}
\end{equation}
To protect the output privacy, client $n$ encrypts $H_n$ using its secret key $SK_n$ as $R_{n}\leftarrow {\bf SE.Enc}(SK_n,H_n)$ and sends $R_n$ to the server. After receiving the ciphertexts from at least $t$ clients (denoted as $C_5\subseteq C_4$), the server forwards received ciphertexts set $\left\{R_v\right\}_{v \in C_5}$ to all clients in $C_5$. The role of clients $C_5$ is to recover the mask added to the models of all clients in $C_4$.

\item {\bf Round 5:} Each client $n \in C_5$ receives the ciphertexts set $\left\{R_v\right\}_{v \in C_5}$ from the server. For each $R_v$, client $n$ decrypts it as $H_{v}\leftarrow {\bf SE.Dec}(SK_v,R_v)$, where $H_v=\left(b_v, \left\{b_{u,v}\right\}_{u \in C_4 \backslash\left\{v\right\}},\left\{N_{u,v}^{SK}\right\}_{u \in C_2 \backslash C_4}\right)$. Using the secret sharing protocol, client $n$ reconstructs the secrets as:
\begin{equation}
\begin{aligned}
\left\{
        \begin{array}{ll}
            b_u\leftarrow {\bf SS.Recon}\left(t,\left\{b_{u,v}\right\}_{v\in C_5}\right),&u \in C_4 \backslash C_5\\
            N_u^{SK}\leftarrow {\bf SS.Recon}\left(t,\left\{N_{u,v}^{SK}\right\}_{v\in C_5}\right),&u \in C_2 \backslash C_4
        \end{array}
\right.
\end{aligned}
\end{equation}
For the clients $u \in C_2 \backslash C_4$ and $v \in C_4$, client $n$ computes their agreed mask $s_{u,v}\leftarrow {\bf KA.Agree} \left(N_u^{SK},N_v^{PK}\right)$ via the Diffie-Hellman key agreement protocol. Next, each client computes the final aggregated output vector ${\bf g}$ and inner product value $Com$ as
\begin{flalign}
\label{remove-mask}
\begin{split}
   {\bf g} = {\bf\tilde g} &- \sum_{u\in C_3\backslash C_4}{\bf\tilde g}_u -\sum_{u\in C_4}{\rm\bf PRG}\left(b_u\right) \\
   &- \sum_{u\in C_4,v\in C_2\backslash C_4:u<v}{\rm\bf PRG}\left(s_{u,v}\right)\\
        &+ \sum_{u\in C_4,v\in C_2\backslash C_4:u>v}{\rm\bf PRG}\left(s_{v,u}\right) = \sum_{u\in C_4}{\bf g_u}.
\end{split}&
\end{flalign}

\begin{flalign}
\label{com-remove-mask}
\begin{split}
   Com &= \sum_{u\in C_4}Com_u - \sum_{u\in C_4}{PRG}\left(b_u\right)\\
   &- \sum_{u\in C_4,v\in C_3\backslash C_4:u<v}{PRG}\left(s_{u,v}\right)\\
    &+ \sum_{u\in C_4,v\in C_3\backslash C_4:u>v}{PRG}\left(s_{v,u}\right)
    = \sum_{u\in C_4}{\bf V}\cdot {\bf g_u}.
\end{split}&
\end{flalign}
Finally, each client $n \in C_5$ performs the output integrity verification by checking whether the equality $Com = {\bf V}\cdot {\bf g}$ holds. If the equality holds, the client accepts ${\bf g}$ as the valid aggregated result; otherwise, it detects malicious server behavior.

\end{itemize}

\begin{figure*} 
  \centering  
  
  \begin{tcolorbox}[
    width=\linewidth,  
    colframe=black,    
    colback=white,     
    boxsep=4pt,        
    left=5pt,          
    right=5pt,          
    fontupper=\small
  ]
  \setlist*[itemize]{leftmargin=*, noitemsep, topsep=0pt, parsep=0pt}

  \textbf{Setup Phase}: All parties are given the security parameter $k$, the number of clients $N$ and a threshold value $t$, honestly generated $pp\leftarrow {\rm\bf KA.Param}(k)$.

  \textbf{Round 0} (Initialization) \textbf{Client $n$:}
  
  \begin{itemize}
  \item Generate key pairs $\left(N_n^{PK},N_n^{SK}\right)\leftarrow {\rm\bf KA.Gen}(pp), \left(P_n^{PK},P_n^{SK}\right)\leftarrow {\rm\bf KA.Gen}(pp), \left(K_n^{PK},K_n^{SK}\right)\leftarrow {\rm\bf SIG.Gen}(k),  \left(SK_n\right)\leftarrow {\bf SE.Gen}(k)$.
  \item Send $\left(N_n^{PK},P_n^{PK},K_n^{PK}\right)$ to the server (through the private authenticated channel) 
and move to next round.
  \end{itemize}
  
  \textbf{Server:}
  \begin{itemize}
  \item Collect messages from at least $t$ clients in the previous round (denote with $C_1$ this set of clients). Otherwise, abort.
  \item Broadcast to all clients in $C_1$ the list $\left\{\left(m,N_m^{PK},P_m^{PK},K_m^{PK} \right)\right\}_{m\in C_1}$ and move to next round.
  \end{itemize}
  
  \textbf{Round 1} (Key Sharing) \textbf{Client $n$:}
  \begin{itemize}
  \item Receive the list $\left\{\left(m,N_m^{PK},P_m^{PK},K_m^{PK}\right)\right\}_{m\in C_1}$ broadcasted by the server. Check whether $|C_1|\geq t$ and all of the key pairs $\left(N_m^{PK},P_m^{PK},K_m^{PK}\right)$ are distinctive. If not, abort.
  \item Select a random number $b_n$. Generate the shares of $b_n$ as $\left\{\left(m,b_{n,m}\right)\right\}_{m\in C_1} \leftarrow {\bf SS.Share}\left(b_n,t,C_1\right)$, where $b_{n,m}$ is the share of client $n$ to client $m$.
  \item Generate the shares of $N_n^{SK}$ as $\left\{\left(m,N_{n,m}^{SK}\right)\right\}_{m\in C_1 }\leftarrow {\bf SS.Share}(N_n^{SK},t,C_1)$, where $N_{n,m}^{SK}$ is the share of client $n$ to $m$.
  \item For each other client $m\in C_1\backslash\left\{n\right\}$, compute $E_{n,m}\leftarrow {\bf SE.Enc}\left(N_{n,m}^{SK} || b_{n,m} ||SK_n,{\bf KA.Agree}\left(P_n^{SK},P_m^{PK}\right)\right)_{m\in C_1\backslash\left\{n\right\}}$.
  \item Send $\left\{\left(n,m,E_{n,m}\right)\right\}_{m\in C_1\backslash\left\{n\right\}}$ to the server.
  \end{itemize}

  \textbf{Server:}
  \begin{itemize}
  \item Collect messages from at least $t$ clients (denote with $C_2\in C_1$ this set of clients). Otherwise, abort.
  \item Send to each client $n\in C_2$ the set $\left\{\left(m,n,E_{m,n}\right)\right\}_{m\in C_2\backslash\left\{n\right\}}$.
  \end{itemize}

  \textbf{Round 2} (Masking Input Vector) \textbf{Client $n$:}
  \begin{itemize}
  \item Receive from the server the list of $\left\{\left(m,n,E_{m,n}\right)\right\}_{m\in C_2\backslash\left\{n\right\}}$. Check whether $C_2\in C_1$ and $|C_2|\geq t$. If not, abort.
  \item For each other client $m\in C_2\backslash\left\{n\right\}$, compute $S_{n,m}\leftarrow {\bf KA.Agree}\left(N_n^{SK},N_m^{PK}\right)$.
  \item Compute the masked input vector
${\bf\tilde g}_n = {\bf g}_n + {\rm\bf PRG}\left(b_n\right) + \sum_{m\in C_2:n<m}{\rm\bf PRG}\left(s_{n,m}\right) - \sum_{m\in C_2:n>m}{\rm\bf PRG}\left(s_{n,m}\right)$.
  \item Generate the signature $\sigma_n^g\leftarrow {\bf SIG.Sign}\left(K_n^{SK},{\bf\tilde g}_n\right)$ and send $\left({\bf\tilde g}_n,\sigma_n^g\right)$ to the server.
  \end{itemize}

  \textbf{Server:}
  \begin{itemize}
    \item Collect messages from at least $t$ clients (denote with $C_3\in C_2$ this set of clients). Otherwise, abort.
    \item Compute ${\bf\tilde g} = \sum_{m\in C_3}{\bf\tilde g}_m$. Send ${\bf\tilde g}$ and the list $C_3$ to each client $n\in C_3$.
  \end{itemize}

  \textbf{Round 3} (Verification Generation) \textbf{Client $n$:}
  \begin{itemize}
    \item Receive ${\bf\tilde g}$ from the server. Check whether $C_3\in C_2,|C_3|\geq t$. If not, abort.
    \item Compute ${\rm\bf V} = {\rm\bf PRG}\left(Hash\left({\bf\tilde g}\right)\right)$ and $Com_n = {\rm\bf V}\cdot{\rm\bf g_n} +{PRG}(b_n) +\sum_{m\in C_3:n<m}{PRG}\left(s_{n,m}\right) - \sum_{m\in C_3:n>m}{PRG}\left(s_{n,m}\right)$.
    \item Generate the signature $\sigma_n^c\leftarrow {\bf SIG.Sign}\left(K_n^{SK},Com_n\right)$ and send $\left(Com_n,\sigma_n^c\right)$ to the server.
  \end{itemize}

  \textbf{Server:}
  \begin{itemize}
    \item Collect messages from at least $t$ clients (denote with $C_4\in C_3$ this set of clients). Otherwise, abort.
    \item Broadcast $\left(\left\{\left(m,Com_m,\sigma_m^c\right)\right\}_{m\in C_4},\left\{\left(m,{\bf\tilde g}_m,\sigma_m^g\right)\right\}_{m\in C_3\backslash C_4}\right)$ to all clients in $C_4$.
  \end{itemize}
    
  \textbf{Round 4} (Secret Shares Encryption) \textbf{Client $n$:}
  \begin{itemize}
    \item Receive messages broadcasted by the server. Check whether $C_4\in C_3,|C_4|\geq t$ and each signature $\sigma_m$ is correct. If not, abort.
    \item For each other client $m\in C_2\backslash \left\{n\right\}$, decrypt the ciphertext as $N_{m,n}^{SK} ||b_{m,n} ||SK_m\leftarrow {\bf SE.Dec}\left(E_{m,n}, {\bf KA.Agree}\left(P_n^{SK},P_m^{PK}\right)\right)$.
   \item Collect $H_n=\left(b_n, \left\{b_{m,n}\right\}_{m \in C_4 \backslash\left\{n\right\}},\left\{N_{m,n}^{SK}\right\}_{m \in C_2 \backslash C_4}\right)$.
   \item Encrypt $H_n$ as $R_{n}\leftarrow {\bf SE.Enc}(SK_n,H_n)$ and send $R_n$ to the server.
  \end{itemize}

  \textbf{Server:}
  \begin{itemize}
    \item Collect messages from at least $t$ clients (denote with $C_5\in C_4$ this set of clients). Otherwise, abort.
    \item Broadcast $\left\{R_u\right\}_{u \in C_5}$ to all clients in $C_5$.
  \end{itemize}

  \textbf{Round 5} (Local Aggregation) \textbf{Client $n$:}
  \begin{itemize}
    \item Receive $\left\{R_u\right\}_{u \in C_5}$ broadcasted by the server. Check whether $C_5\in C_4$ and $|C_5|\geq t$. If not, abort.
    \item For each other client $u\in C_5$, decrypt each $R_u$ as $H_u\leftarrow {\bf SE.Dec}\left(SK_u,R_u\right)$.
    \item For each other client $u \in C_4 \backslash C_5$, compute $b_u\leftarrow {\bf SS.Recon}\left(t,\left\{b_{u,v}\right\}_{v\in C_5}\right)$.
    \item For each other client $u \in C_2 \backslash C_4$, compute $N_u^{SK}\leftarrow {\bf SS.Recon}\left(t,\left\{N_{u,v}^{SK}\right\}_{v\in C_5}\right)$.
    \item Compute ${\bf PRG}\left(s_{u,v}\right)\leftarrow {\bf PRG}\left({\bf KA.Agree}\left({N_u^{SK},N_v^{PK}}_{u\in C_2\backslash C_4,v\in C_5}\right)\right)$ and ${\bf PRG}\left(b_u\right)_{u\in C_4}$.
    \item Calculate the aggregation result as
    \begin{equation}
    \begin{aligned}
        &{\bf g} = \sum_{u\in C_4}{\bf g}_u = {\bf\tilde g} - \sum_{u\in C_3\backslash C_4}{\bf\tilde g}_u - \sum_{u\in C_4}{\rm\bf PRG}\left(b_u\right) - \sum_{u\in C_4,v\in C_2\backslash C_4:u<v}{\rm\bf PRG}\left(s_{u,v}\right) + \sum_{u\in C_4,v\in C_2\backslash C_4:u>v}{\rm\bf PRG}\left(s_{v,u}\right)\\
        &Com = \sum_{u\in C_4}Com_u - \sum_{u\in C_4}{PRG}\left(b_u\right) - \sum_{u\in C_4,v\in C_3\backslash C_4:u<v}{PRG}\left(s_{u,v}\right) + \sum_{u\in C_4,v\in C_3\backslash C_4:u>v}{PRG}\left(s_{v,u}\right)
    \end{aligned}
    \end{equation}
    \item Check whether $Com == {\bf V}\cdot {\bf g}$.
  \end{itemize}
  
  \end{tcolorbox}

  \caption{Detailed EIFL.}
  \label{fig:fig3}
\end{figure*}

\subsection{Complexity Analysis}
We assume a single server and $N$ total clients, and all clients have model vectors with a length of $M$. We use $N_d$ and $N_v$ to denote the number of dropped clients in the non-verification phase and verification phase, respectively. To facilitate analysis, we assume that the dropout rate of clients is fixed. The result is shown in Table \ref{Complexity Analysis}.

\subsubsection{Computation Complexity of Client} $O(MN^2)$. The computation costs of each client can be broken up into 9 parts: (1) Perform a total of $2N$ key agreements, which is $O\left(N\right)$ time, (2) Create $t$-out-of-$N$ Shamir's secret shares of $b_n$ and $N_n^{SK}$, which is $O\left(N^2\right)$ time, (3) Generate values ${\rm\bf PRG}(b_n)$ and ${\rm\bf PRG}(s_{n,m})$ for every other client $m$ for each entry in the input gradient vector, which takes $O\left(MN\right)$ time in total, (4) Generate the vector ${\bf U}$ by using ${Hash}$ function and ${\bf PRG}$ and compute the masked inner product value $Com_n$ in {\bf Round 3}, which takes $O\left(M+N\right)$ time in total, (5) Perform $N-1$ and $1$ symmetric encryption in {\bf Round 1} and {\bf Round 4} respectively, which takes $O\left(N\right)$ time in total, (6) Perform $N-1$ and $\left(N-N_d-N_V-1\right)$ symmetric decryption in {\bf Round 4} and {\bf Round 5} respectively, which takes $O(N)$ time in total, (7) Perform one signature in {\bf Round 3} and $N_d$ verification in {\bf Round 4}, respectively, which takes $O\left(M+N\right)$ time in total, (8) Reconstruct $N_d$ $t$-out-of-$N$ Shamir's secrets (one for each dropped client) in {\bf Round 5}, which takes $O(N^2)$ time, (9) Generate and remove the appropriate ${\rm\bf PRG}(b_n)$ and ${\rm\bf PRG}(s_{n,m})$ values from the intermediate result vector ${\bf\tilde g}$ in {\bf Round 5}, which takes time $O\left(MN^2\right)$ in the worst case. Overall, each client’s computation complexity is $O(MN^2)$.

\subsubsection{Communication Complexity of Client} $O(N^2+MN)$. The communication cost of each client can be split into 5 parts: (1) Send $3$ and receive $3(N-1)$ public keys in {\bf Round 0}, which takes $O(N)$ communication, (2) Send $3(N-1)$ and receive $3(N-1)$ ciphertexts in {\bf Round 1}, which takes $O(N)$ communication, (3) Send $1$ masked input-signature pair and receive $1$ intermediate result vector in {\bf Round 2}, which is $O(M)$, (4) Send $1$ inner product value-signature pair, receive $N-N_d-N_v-1$ inner product value-signature pairs and $N_v$ masked input-signature pairs in {\bf Round 3}, which is $O(MN)$ and (5) Send $N$ and receive $N(N- N_d - N_v - 1)$ ciphertexts in {\bf Round 4}, which is $O(N^2)$. Hence, the total communication cost is $O(N^2+MN)$.

\subsubsection{Computation Complexity of Server} $O(MN)$. The server generate the intermediate result vector using $N-N_d$ clients' masked input vectors in {\bf Round 2}, which takes $O\left(MN\right)$ time.

\subsubsection{Communication Complexity of Server} $O(N^3+MN^2)$. The communication complexity of the server can be broken up into 5 parts: (1) Receive $3N$ and send $3N\left(N-1\right)$ public keys in {\bf Round 0}, (2) Receive $3N(N-1)$ and send $3N\left(N-1\right)$ ciphertexts in {\bf Round 1}, (3) Receive $N-N_d$ masked input-signature pairs and send $N-N_d$ intermediate result vectors in {\bf Round 2}, (4) Receive $N-N_d-N_v$ scalar value-signature pairs, send $\left(N-N_d-N_v\right)\left(N - N_d - N_v - 1\right)$ scalar value-signature pairs and $N_v\left(N - N_d - N_v\right)$ masked input-signature pairs in {\bf Round 3}, (5) Receive $N\left(N-N_d-N_v\right)$ and send $N\left(N-N_d-N_v\right)(N- N_d - N_v - 1)$ ciphertexts in {\bf Round 4}. 
Hence, the communication complexity of the server is $O(N^3+MN^2)$ in total.

\begin{table}
    \centering
    \small
    \caption{Complexity analysis of our protocol.}
  \addtolength{\tabcolsep}{2pt}
  \begin{tabular}{lccc}
    \toprule
    & Client   & Server   \\
    \midrule
    Computation                     & $O(MN^2)$    &  $O(MN)$ \\
    Communication                   & $O(N^2+MN)$     & $O(N^3+MN^2)$  \\
    \bottomrule
  \end{tabular}
    \label{Complexity Analysis}
    \vspace{-2mm}
\end{table}

\section{Security Analysis of EIFL}
\label{Analysis}

In this section, we provide the formal definition of our EIFL and analyze how our EIFL guarantees the output privacy, and output integrity under the corresponding threat model.

We use $S$ to represent the server, and $C_k \subseteq C$ to denote the clients who correctly upload data to the server at round $k-1$. Note that we have $C_5\subseteq C_4\subseteq C_3\subseteq C_2\subseteq C_1\subseteq C$. $C_k\backslash C_{k+1}$ denotes the clients who have sent data to the server in Round $k-1$, but abort before uploading data in Round $k$. We denote the input vector of each client $n$ with ${\bf g}_n$, and ${\bf g}_{C}=\left\{{\bf g}_n\right\}_{n\in C}$ indicates inputs of the clients $C$.  

\subsection{Output Privacy}
In EIFL, the view of the server $S$ is defined as its internal state (including the intermediate result and randomness) and all messages it received from clients. The view of the server can be seen as a random variable $REAL_S^{C,t,k}\left({\bf g}_C\right)$, where $t$ and $k$ indicate the threshold and security parameter used in our protocol, respectively.
We say that the server learns nothing if its view can be simulated by a simulator without any secret internal state of the clients (e.g., the input of a client). 
Below, we provide the formal definition of EIFL.

\begin{theorem}[Output privacy]
\label{output privacy}
{For all $t,  k, C$, and a corrupted server $S$, there exists a Probabilistic Polynomial Time (PPT) simulator $SIM$ whose output is indistinguishable from the output of $REAL_U^{S,t,k}$, i.e., 
\begin{small}
\begin{align*}
REAL_S^{C,t,k} \left({\bf g}_C\right) \approx  SIM_S^{C,t,k}\left({\bf g}_C\right) \nonumber
\end{align*}
\end{small}%
}
\end{theorem}

\begin{proof} We use the hybrid argument~\cite{DBLP:cite31} to prove the theorem. The key step to protect output privacy is {\bf Round 4}, in which each client uploads the ciphertexts that need to be decrypted for calculating the clients' model aggregation. Our main proof idea is to replace these ciphertexts and use the hybrid argument to prove that the ciphertexts before and after replacement is computationally indistinguishable to the server. Since we exclude the involvement of the clients, the view of the server does not depend on the real aggregation output of clients.

${\rm Hyb_0}$: The \emph{SIM} random variable is distributed exactly as $REAL$, the view of the server in a real execution of the protocol.

${\rm Hyb_1}$: In this hybrid, the simulator changes the behavior of all clients $n$, where $n\in C_2$. Specifically, for each client $n$, a uniformly random number $v_{n,m}$ is selected to replace the shared key KA.Agree$\left(P_n^{SK},P_m^{PK}\right)$ between client $n$ and client $m$ in the same set, and to perform the function of encryption and decryption. The DDH assumption \cite{DBLP:cite19} ensures that this hybrid is indistinguishable from real protocol.

${\rm Hyb_2}$: In this hybrid, the simulator replaces all encrypted data (i.e., the encrypted $SK$ and encrypted shares of $b_n$ and $N_n^{SK}$) sent by all clients to other clients in {\bf Round 1} with encrypted random values (e.g., 0, with appropriate length). Since the simulator just changes the content of ciphertext, IND-CPA security of the Symmetric Encryption (SE) scheme \cite{DBLP:cite5}, \cite{DBLP:cite20} guarantees this hybrid is indistinguishable from the previous one. 

${\rm Hyb_3}$: In this hybrid, the simulator replaces all encrypted data (i.e., the encrypted its $b_n$ and encrypted shares of $b_m$ and $N^{SK}_m$ of all other clients $m$) sent in {\bf Round 4} with encrypted shares of random values (e.g., 0, with appropriate length). Since the simulator only changes the contents of the ciphertexts, IND-CPA security of our Symmetric Encryption (SE) scheme \cite{DBLP:cite5}, \cite{DBLP:cite20} guarantees that the view of this hybrid is indistinguishable from the previous one. 

Therefore, the simulator has already completed the simulation since SIM successfully simulates REAL without knowing anything about the aggregation outputs. Based on the hybrid 0 to 3, we can proves that the output of the simulator is computationally indistinguishable from the output of REAL, completing the proof.
\end{proof}

\subsection{Output Integrity}
\begin{theorem}[Output integrity]
\label{output integrity}
For any Probabilistic Polynomial Time (PPT) malicious aggregation server, the protocol output ${\bf g}$ is exactly the sum of the local gradients of all surviving clients in $C_4$, if and only if the verification equation $Com == {\bf V}\cdot {\bf g}$ holds. Formally, $$Pr\left[{\bf g}={\bf g}_{Integrity}\right] \geq 1-negl(k),$$ where:
\begin{itemize}
    \item ${\bf g}_{Integrity} = \sum_{n\in C_4}{\bf g}_n$ denotes the valid global aggregated model of surviving clients in $C_4$;
    \item $Com = \sum_{n\in C_4}Com_n$ denotes the aggregated inner product values from surviving clients;
    \item ${\bf V}$ denotes the random vector generated via pseudorandom generator ($\bf PRG$);
    \item $negl(k)$ denotes a negligible function in the security parameter $k$.
\end{itemize}
\end{theorem}
\begin{proof}
We prove this theorem under two standard cryptographic assumptions:
\begin{itemize}
    \item the $Hash$ function is collision-resistant (CRHF);
    \item the pseudo-random generator $\bf PRG$ is cryptographically secure (i.e., its output is computationally indistinguishable from a uniformly random vector for any PPT adversary).
\end{itemize}
We split the proof into two core parts: {\bf completeness} (honest execution always passes verification) and {\bf soundness} (a PPT malicious server cannot forge a verifiable invalid output, except with negligible probability).

\subsubsection{Completeness (Honest Execution Passes Verification)}
We first prove that the honest protocol execution will always produce an output that satisfies the verification equation.
\begin{itemize}
    \item Each client $n \in C_2$ computes the masked input vector ${\bf\tilde g}_n$ according to equation \ref{local-mask}, and uploads it to the server. The server honestly sums all these masked vectors of clients in $C_3$ to get the valid intermediate result ${\bf\tilde g} = \sum_{m\in C_3}{\bf\tilde g}_m$.
    \item Each client $n \in C_3$ computes the masked inner product value $Com_n$ according to equation \ref{com-mask}. For the clients that dropped out during uploading masked inner product values in {\bf Round 3}, surviving clients offset these clients' masked input vectors, and sum all the original input vectors of clients in $C_4$ as ${\bf g} =\sum_{u\in C_4}{\bf g_u}$ according to equation \ref{remove-mask}.
    \item Only the masked inner product values of clients in $C_4$ are uploaded to the server. According to equation \ref{com-remove-mask}, clients can calculate the sum of the real inner product values $Com = \sum_{u\in C_4}{\bf V}\cdot{\rm\bf g_u}$ .
\end{itemize}
By the linearity of the inner product, we have:
{
\begin{align*}
Com = \sum_{u\in C_4}{\bf V}\cdot{\rm\bf g_u} = {\bf V}\cdot\sum_{u\in C_4}{\rm\bf g_u} = {\bf V}\cdot{\rm\bf g}.
\nonumber
\end{align*}
} 
This proves that the honest protocol output always satisfies the verification equation, completing the proof of completeness.

\subsubsection{Soundness (Unforgeability of Invalid Outputs)} We now prove that no PPT malicious server can forge an invalid intermediate result ${\bf\tilde g}_f \neq {\bf\tilde g}$ that passes the client-side verification, except with negligible probability. First, the verification equation for a forged intermediate result is:
\begin{equation}
    \begin{aligned}
        {\bf V}_f\cdot{\bf g}_f = {\bf V}_f\cdot{\rm\bf g},
    \end{aligned}
\end{equation}
where the vector ${\bf V}_f ={\bf PRG}\left(Hash\left({\bf\tilde g}_f\right)\right)$ is uniquely determined by the forged intermediate result ${\bf\tilde g}_f$, ${\bf g}_f$ is the incorrect output derived from the forged intermediate result ${\bf\tilde g}_f$, and ${\bf g}$ is the correct output derived from the valid intermediate result ${\bf\tilde g}$. According to equation \ref{remove-mask}, we note that the mask removal process is independent of the server’s intermediate result: the mask term is computed solely from the clients’ pre-shared secret seeds and the set of surviving clients, and is identical for both the honest and forged intermediate results. For notational convenience, we denote the mask term as:
\begin{align*}
{\bf mask} = - \sum_{u\in C_3\backslash C_4}{\bf g_u^-} -\sum_{u\in C_4}{\rm\bf PRG}\left(b_u\right) \\
- \sum_{u\in C_4,v\in C_2\backslash C_4:u<v}{\rm\bf PRG}\left(s_{u,v}\right)\\
        + \sum_{u\in C_4,v\in C_2\backslash C_4:u>v}{\rm\bf PRG}\left(s_{v,u}\right)
\nonumber
\end{align*}
Therefore, the incorrect output vector ${\bf g}_f$ and the correct output vector ${\bf g}$ can be expressed as
\begin{align*}
    {\bf g}_f={\bf\tilde g}_f+{\bf mask},{\bf g} = {\bf\tilde g}+{\bf mask},
\end{align*}
and the verification equation can be expressed as
\begin{align*}
{\bf V}_f\cdot({\bf\tilde g}_f+{\bf mask}) = {\bf V}_f\cdot({\bf\tilde g}+{\bf mask}),
\end{align*}
Substituting into the verification equation and canceling the common ${\bf V}_f \cdot{\bf mask}$ term via inner product linearity, the adversary’s forgery goal reduces to finding ${\bf\tilde g}_f \neq {\bf\tilde g}$ that satisfies:
\begin{equation}
\begin{aligned}
{\bf V}_f\cdot{\bf\tilde g}_f = {\bf V}_f\cdot{\bf\tilde g}
\end{aligned}
\label{verification-equation}
\end{equation}
We observe that ${\bf V}_f$ is generated from ${\bf\tilde g}_f$, and ${\bf\tilde g}$ is a constant vector that can be computed by the server, which means that there is only one unknown vector in the equation \ref{verification-equation}. However, under our cryptographic assumptions, ${\bf V}_f$ is uniquely determined by ${\bf\tilde g}_f$ and is computationally indistinguishable from a uniform random vector. A PPT adversary cannot manipulate the linear relationship between ${\bf V}_f$ and ${\bf\tilde g}_f$, nor express ${\bf V}_f$ as a closed-form function of ${\bf\tilde g}_f$. The only feasible attack strategy is exhaustive search over the vector space, which has a super-polynomial time complexity of $O(2^{lm})$ (where $l$ is the bit length of each vector element and $M$ is the vector dimension), and is thus infeasible for any PPT adversary. The success probability of the forgery attack is therefore negligible.

In conclusion, honest execution always passes verification, and any forged invalid output will be detected by clients with overwhelming probability. This completes the proof.
\end{proof}

\section{Evaluation}
\label{Evaluation}
In this section, we evaluate the performance of our EIFL in terms of effectiveness, computation overhead, and communication overhead.
\begin{table*}
  \small
  \centering
  \caption{Performance with Different Numbers of Clients $N$ and Model Dimensions $M$}
  \addtolength{\tabcolsep}{-1pt}
  
  \begin{tabular}{lccccccccc}
    \toprule
     &$M$ & & Round 0 & Round 1 & Round 2 & Round 3  & Round 4 & Round 5 & Communication \\
     \midrule
    \multirow{10}{*}{$N=20$} & \multirow{2}{*}{$2*10^4$}
    &Client  & 289 ms & 11 ms & 202 ms & \textcolor{blue}{36} ms & 0.2+\textcolor{blue}{9}+\textcolor{orange}{0.2} ms & 162+\textcolor{blue}{6}+\textcolor{orange}{0.5} ms  & 135.81+\textcolor{blue}{1.56} KB\\
    & &Server   & - & - & - & 0.4 ms & - & - & 2.68 MB \\
    \cmidrule(lr){2-10}
    & \multirow{2}{*}{$4*10^4$}
    &Client   & 281 ms & 12 ms & 372 ms & \textcolor{blue}{47} ms & 0.2+\textcolor{blue}{9}+\textcolor{orange}{0.3} ms & 318+\textcolor{blue}{13}+\textcolor{orange}{0.5} ms  & 253+\textcolor{blue}{1.56} KB\\
    & &Server    & - & - & - & 0.9 ms & - & - & 4.97 MB \\
    \cmidrule(lr){2-10}
    & \multirow{2}{*}{$6*10^4$}
    &Client   & 286 ms & 12 ms & 538 ms & \textcolor{blue}{62} ms & 0.2+\textcolor{blue}{9}+\textcolor{orange}{0.3} ms & 477+\textcolor{blue}{17}+\textcolor{orange}{0.5} ms  & 370.19+\textcolor{blue}{1.56} KB\\
    & &Server    & - & - & - & 1.2 ms & - & - & 7.26 MB \\
    \cmidrule(lr){2-10}
    & \multirow{2}{*}{$8*10^4$}
    &Client   & 286 ms & 12 ms & 701 ms & \textcolor{blue}{75} ms & 0.2+\textcolor{blue}{9}+\textcolor{orange}{0.3} ms & 649+\textcolor{blue}{22}+\textcolor{orange}{0.5} ms  & 487.38+\textcolor{blue}{1.56} KB\\
    & &Server    & - & - & - & 1.5 ms & - & - & 9.55 MB \\
    \cmidrule(lr){2-10}
    & \multirow{2}{*}{$10^5$}
    &Client   & 275 ms & 12 ms & 892 ms & \textcolor{blue}{90} ms & 0.2+\textcolor{blue}{9}+\textcolor{orange}{0.3} ms & 811+\textcolor{blue}{29}+\textcolor{orange}{0.5} ms  & 604.56+\textcolor{blue}{1.56} KB\\
    & &Server    & - & - & - & 2 ms & - & - & 11.84 MB \\
    
    \midrule
    \multirow{10}{*}{$N=100$} & \multirow{2}{*}{$2*10^4$}
    &Client  & 283 ms & 61 ms & 849 ms & \textcolor{blue}{35} ms & 1+\textcolor{blue}{46}+\textcolor{orange}{1} ms & 797+\textcolor{blue}{6}+\textcolor{orange}{2} ms  & 460.81+\textcolor{blue}{7.81} KB\\
    & &Server   & - & - & - & 2 ms & - & - & 45.76 MB \\
    \cmidrule(lr){2-10}
    & \multirow{2}{*}{$4*10^4$}
    &Client   & 276 ms & 61 ms & 1636 ms & \textcolor{blue}{48} ms & 1+\textcolor{blue}{46}+\textcolor{orange}{1} ms & 1594+\textcolor{blue}{12}+\textcolor{orange}{2} ms  & 578+\textcolor{blue}{7.81} KB\\
    & &Server    & - & - & - & 3.9 ms & - & - & 57.21 MB \\
    \cmidrule(lr){2-10}
    & \multirow{2}{*}{$6*10^4$}
    &Client   & 286 ms & 60 ms & 2441 ms & \textcolor{blue}{62} ms & 2+\textcolor{blue}{46}+\textcolor{orange}{1} ms & 2386+\textcolor{blue}{17}+\textcolor{orange}{2} ms  & 703+\textcolor{blue}{7.81} KB\\
    & &Server    & - & - & - & 6 ms & - & - & 68.65 MB \\
    \cmidrule(lr){2-10}
    & \multirow{2}{*}{$8*10^4$}
    &Client   & 291 ms & 61 ms & 3275 ms & \textcolor{blue}{76} ms & 2+\textcolor{blue}{45}+\textcolor{orange}{1} ms & 3309+\textcolor{blue}{23}+\textcolor{orange}{2} ms  & 812.38+\textcolor{blue}{7.81} KB\\
    & &Server    & - & - & - & 8 ms & - & - & 80.10 MB \\
    \cmidrule(lr){2-10}
    & \multirow{2}{*}{$10^5$}
    &Client   & 268 ms & 61 ms & 4035 ms & \textcolor{blue}{89} ms & 2+\textcolor{blue}{46}+\textcolor{orange}{1} ms & 4281+\textcolor{blue}{31}+\textcolor{orange}{2} ms  & 937.38+\textcolor{blue}{7.81} KB\\
    & &Server    & - & - & - & 10 ms & - & - & 91.54 MB \\

    \bottomrule
  \end{tabular}
  
  \label{tbl:running time}
\end{table*}

\subsection{Experimental Setup}
\label{Experimental Setup}
We show the performance of our EIFL based on a prototype written in Python with the following cryptographic settings:
\begin{itemize}
    \item For key agreements protocol, we use elliptic curve Diffie-Hellman~\cite{DBLP:cite18} over the NIST P-256 curve.
    \item For secret sharing, we adopt the standard Shamir’s $t$-out-of-$N$ secret sharing~\cite{DBLP:cite16}.
    \item For one-time signature scheme, we use the ECDSA digital signature scheme based on SHA-256 with 512-bit one-time signature secret key.
    \item For authenticated encryption, we use AES in counter mode with 256-bit secret key
\end{itemize}

We run a single-threaded simulation on a Ubuntu 20.04.4 LTS desktop with Intel(R) Xeon(R) w5-3435X CPU, 192GB memory and an NVIDIA RTX 4090 GPU.


\subsection{Model Accuracy and Convergence Rounds}
We follow the settings of the work \cite{rlr} and use the same model and datasets. Specifically, we use a 5-layer convolutional neural network consisting of about 1.2M parameters with the following architecture: two layers of convolution, followed by a layer of max-pooling, followed by two fully-connected layers with dropout. We evaluate EIFL on three image datasets, i.e., Fashion MNIST \cite{Fashion-MNIST}, CIFAR10 \cite{cifar-10} and Federated EMNIST, where the first two datasets are independent identically distributed (i.i.d.)  and the last one is non-i.i.d. distributed. 
\begin{itemize}
    \item \textbf{Fashion MNIST (i.i.d.).} We give each agent an equal number of samples from the training data via uniform sampling. We set the batch size as 256 and clients' local training epochs as 2.
    \item \textbf{CIFAR10 (i.i.d.).} We give each agent an equal number of samples from the training data via uniform sampling. We set the batch size as 256 and clients' local training epochs as 2.
    \item \textbf{Federated EMNIST (non-i.i.d.).} We use the Federated EMNIST dataset from the LEAF benchmark \cite{feder-mnist}. In this dataset, digits are distributed across 3383 users and each user has possibly a different distribution over digits. We set the batch size as 64 and clients' local training epochs as 10.
\end{itemize}
We perform end-to-end federated training across 100 participating clients in all our experiments and the results are shown in Table \ref{model accuracy}. Compared with FedAvg, our protocol exhibits a negligible degradation in classification accuracy (ACC), and requires marginally more convergence rounds (CR). This is mainly because our scheme adopts 16-bit quantization. Overall, the impact of our protocol on model accuracy and the number of convergence rounds is marginal.

\begin{figure*}
     \centering
     \begin{subfigure}[b]{0.24\textwidth}
         \centering
         \includegraphics[scale=0.28]{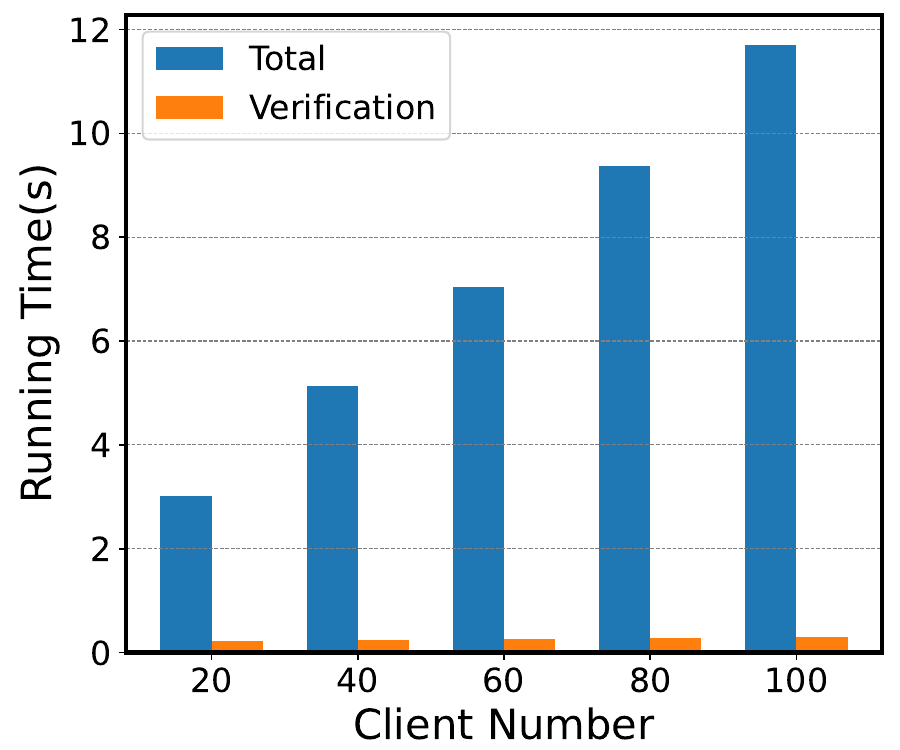}
         \caption{$M=100K$}
         \label{fig:verification-time:a}
     \end{subfigure}
     \begin{subfigure}[b]{0.24\textwidth}
         \centering
         \includegraphics[scale=0.28]{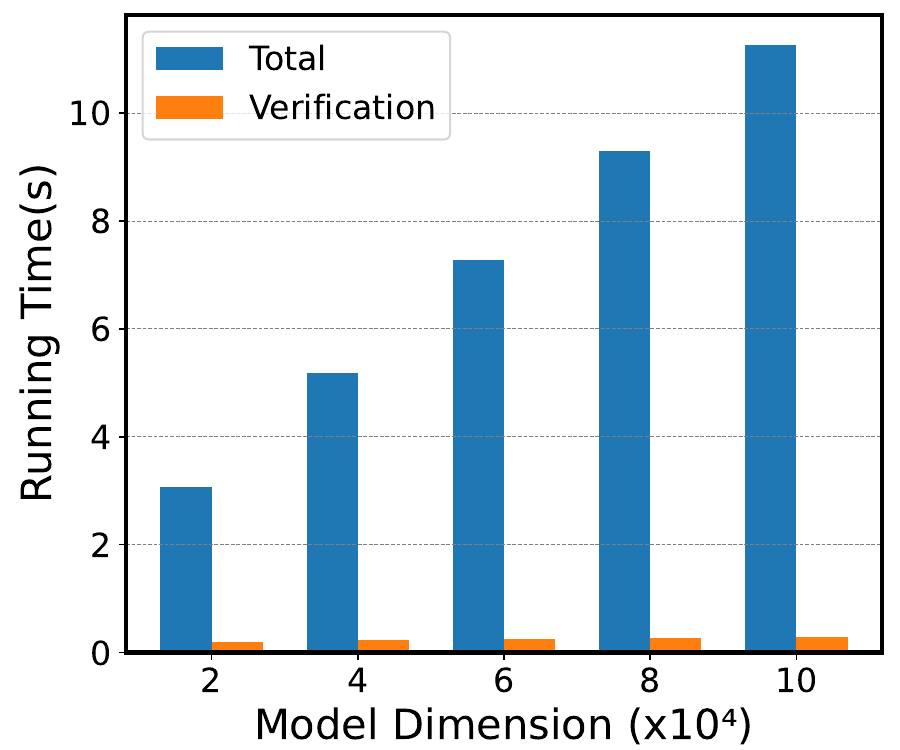}
         \caption{$N=100$}
         \label{fig:verification-time:b}
     \end{subfigure}
     \begin{subfigure}[b]{0.24\textwidth}
         \centering
         \includegraphics[scale=0.28]{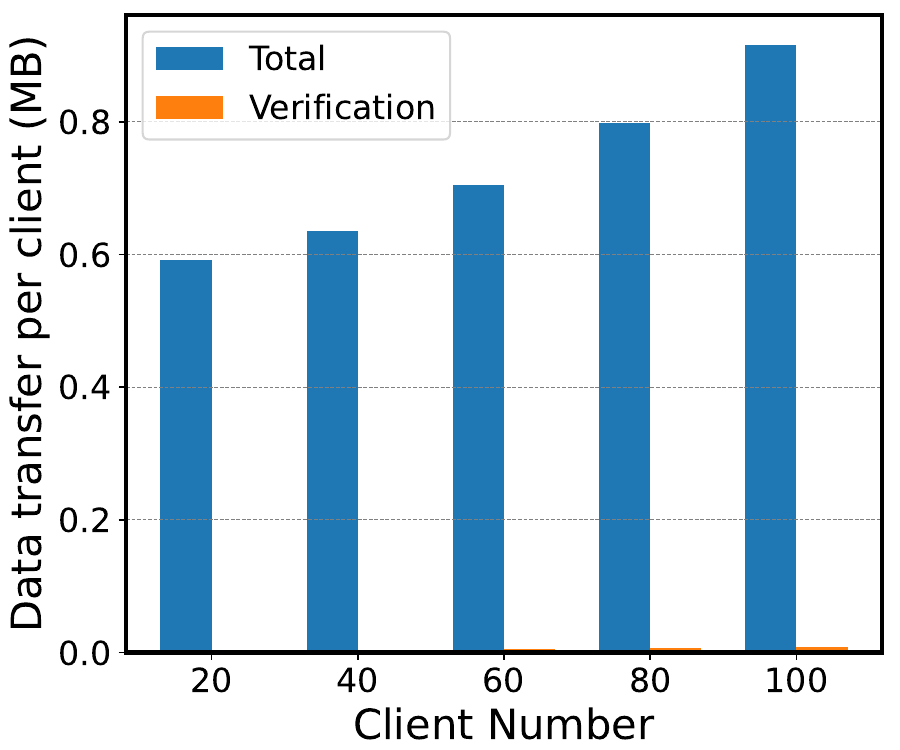}
         \caption{$M=100K$}
         \label{fig:verification-data:c}
     \end{subfigure}
     \begin{subfigure}[b]{0.24\textwidth}
         \centering
         \includegraphics[scale=0.28]{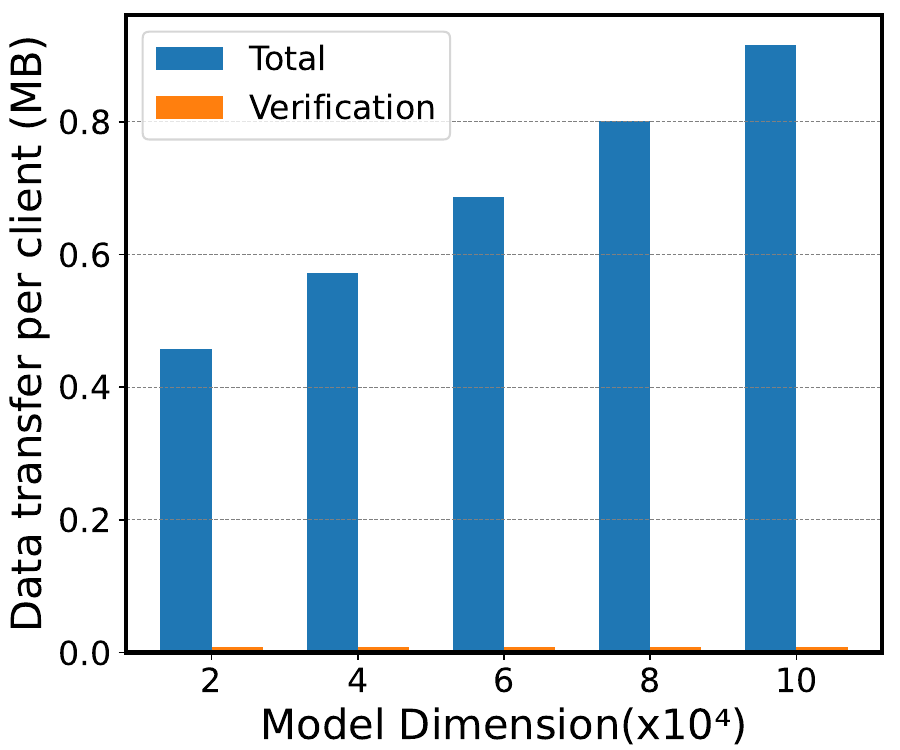}
         \caption{$N=100$}
         \label{fig:verification-data:d}
     \end{subfigure}
     \vspace{-1mm}
        \caption{Comparison of verification overhead and total overhead per client. (a) Running time with different numbers of clients $N$. (b) Running time with different model dimension $M$. (c) Data transfer with different numbers of clients $N$. (d) Data transfer with different model dimension $M$.} 
        \label{fig:verification}
\end{figure*}

\begin{table}
    \centering
    \small
    \caption{Model Accuracy and Convergence Rounds}
  \addtolength{\tabcolsep}{2pt}
  \begin{tabular}{lccccc}
    \toprule
    &\multicolumn{2}{c}{FedAvg} & \multicolumn{2}{c}{EIFL} \\
    \cmidrule(lr){2-3} \cmidrule(lr){4-5}
    & ACC & CR & ACC  & CR  \\
    \midrule
    Fashion MNIST       &0.920 &183 &0.918 &195\\
    Federated EMNIST    & 0.991 &204  &0.974 &226 \\
    CIFAR10                   & 0.790 &152     & 0.782 &168  \\
    \bottomrule
  \end{tabular}
    \label{model accuracy}
    \vspace{-3mm}
\end{table}

\subsection{Performance Analysis of EIFL}
We evaluated the computation overhead and communication overhead of EIFL in terms of runtime and data transfer volume, respectively. Table \ref{tbl:running time} shows the running time of each round and data transfer volume with different numbers of clients $N$ and model dimension $M$. Note that EIFL is designed to run on top of Secure Aggregation \cite{DBLP:cite5} and additionally ensures output integrity and output privacy. For convenience, we use \textcolor{blue}{\bf blue} and \textcolor{orange}{\bf orange} colors to represent the additional overheads incurred by these two aspects, respectively.

\subsubsection{Runtime} As illustrated in Table \ref{tbl:running time}, the additional \textcolor{blue}{overheads} for output integrity in Round 3 and Round 5 mainly increase as the model dimension increases. This is mainly because clients need to generate a random vector with the same dimension as the model and compute its inner product in Round 3, and then verify the inner product of the output vector in Round 5. By contrast, the additional \textcolor{blue}{overhead} in Round 4 stems from verifying the signatures of other clients and increases as the number of clients increases. In terms of output privacy, clients need to encrypt and send messages for output privacy to each other in Round 4, and decrypt the messages received from other clients in Round 5. Therefore, the computation \textcolor{orange}{overheads} in Round 4 and Round 5 increase as the number of clients increases. Fortunately, compared with the total overhead, the additional overheads remain insignificant. Since the server is only responsible for computing the sum of masked input vectors in Round 2, its computation overhead remains very low, even though such overhead increases as the number of clients and the model dimension increase. 

To demonstrate the superior verification performance of our EIFL, we present a comparative analysis of verification running time and total running time in Figures \ref{fig:verification-time:a} and \ref{fig:verification-time:b}. As observed, as the number of clients or the model dimension increases, the verification running time grows slowly and accounts for an increasingly smaller proportion of the total running time.


\begin{figure*}
     \centering
     \begin{subfigure}[b]{0.24\textwidth}
         \centering
         \includegraphics[scale=0.28]{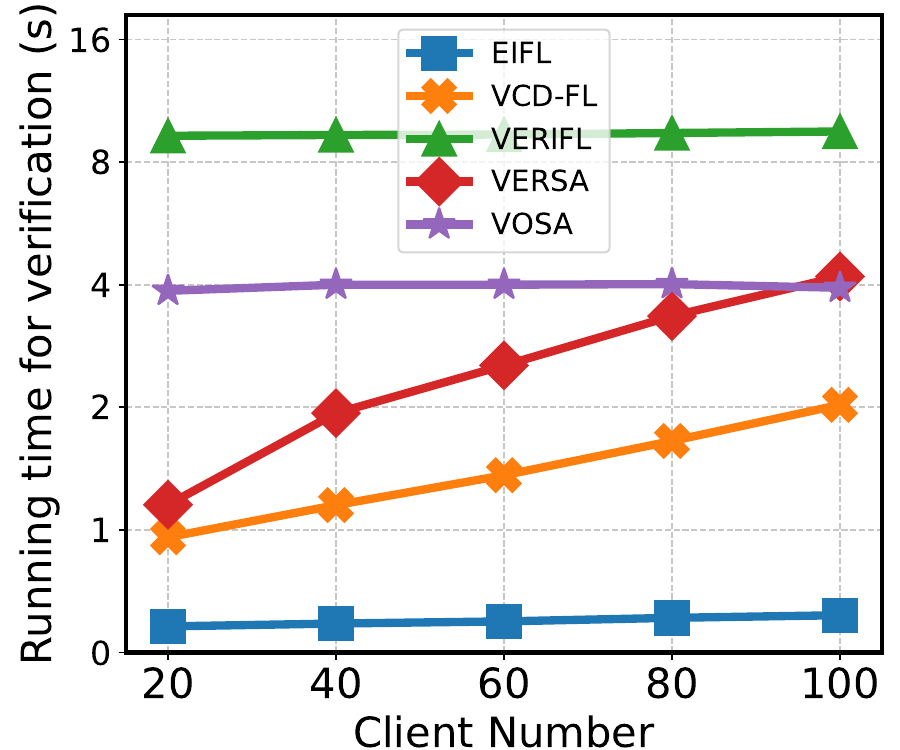}
         \caption{$M=100K$}
         \label{fig:diff_ver_time:a}
     \end{subfigure}
     \begin{subfigure}[b]{0.24\textwidth}
         \centering
         \includegraphics[scale=0.28]{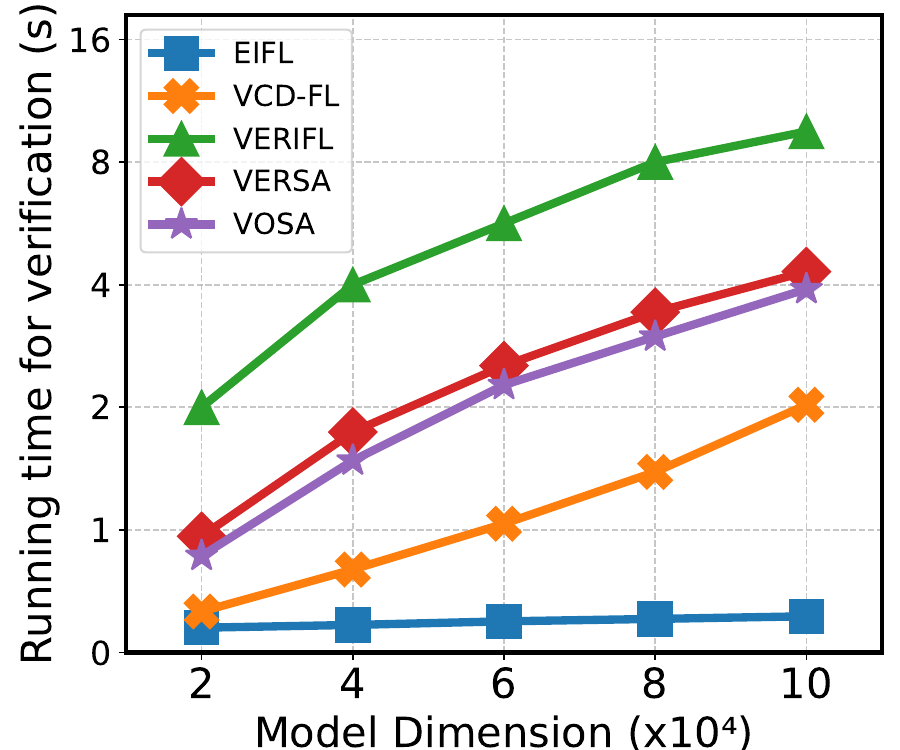}
         \caption{$N=100$}
         \label{fig:diff_ver_time:b}
     \end{subfigure}
     \begin{subfigure}[b]{0.24\textwidth}
         \centering
         \includegraphics[scale=0.28]{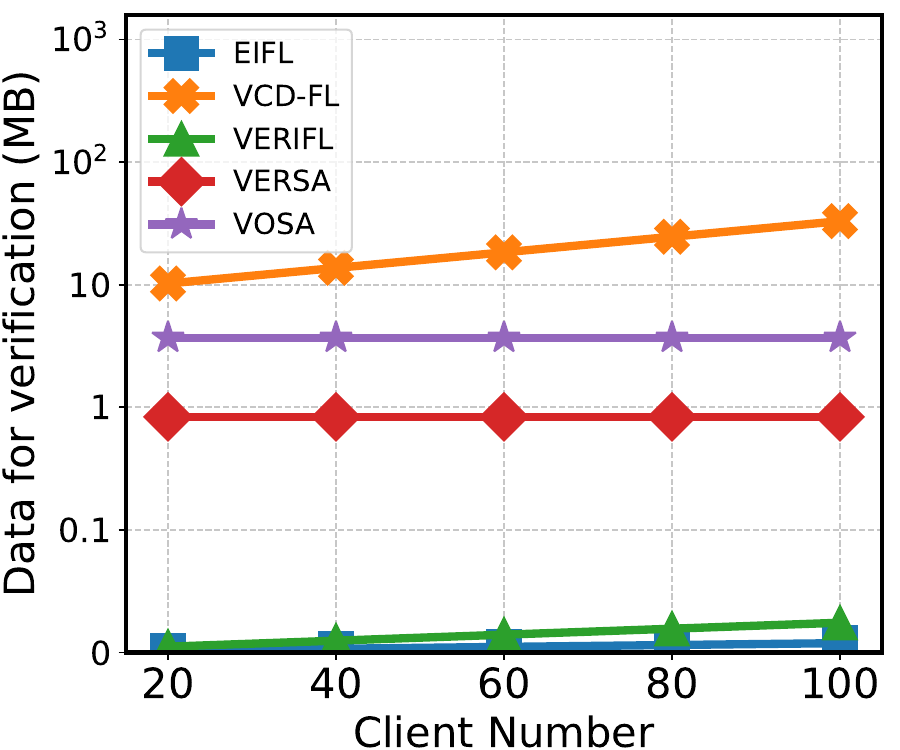}
         \caption{$M=100K$}
         \label{fig:diff_ver_com:c}
     \end{subfigure}
     \begin{subfigure}[b]{0.24\textwidth}
         \centering
         \includegraphics[scale=0.28]{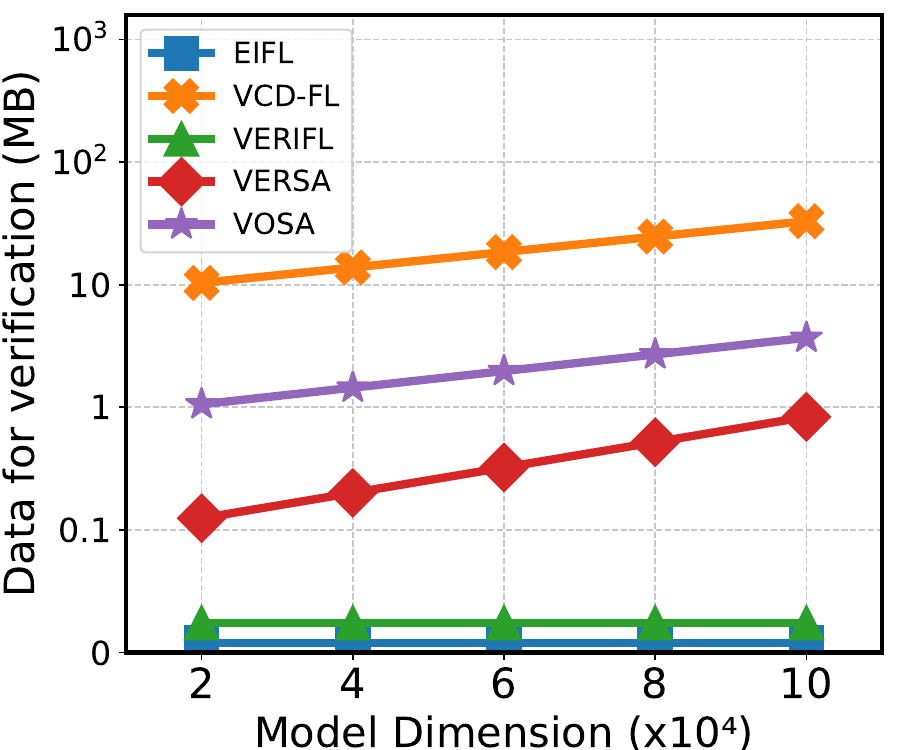}
         \caption{$N=100$}
         \label{fig:diff_ver_com:d}
     \end{subfigure}
     \vspace{-1mm}
        \caption{Comparison of verification performance. (a)Verification time with different numbers of clients $N$. (b) Verification time with different model dimension $M$. (c) Data for verification with different numbers of clients $N$. (d) Data for verification with different model dimension $M$.} 
        \label{fig:diff_ver}
\end{figure*}

\begin{table*}
  \footnotesize
  \centering
  \caption{EIFL vs. VCD-FL\cite{VCD-FL}: Enc/Dec Time with Different Numbers of Clients $N$ and Model Dimensions $M$}
    \vspace{-2mm}
  \addtolength{\tabcolsep}{2pt}
  \begin{tabular}{lccccccccccc}
    \toprule
    &\multicolumn{1}{c}{}  &\multicolumn{5}{c}{$N$ ($M = 100K$)} & \multicolumn{5}{c}{$M$ ($N=100$)} \\
    \cmidrule(lr){3-7} \cmidrule(lr){8-12}
            & &20 & 40 & 60 & 80 & 100  & $2*10^4$ & $4*10^4$ & $6*10^4$ & $8*10^4$ & $10^5$ \\
    \midrule
    \multirow{2}{*}{Encryption (ms)} &EIFL &0.26 & 0.52 & 0.78 & 1.05 & 1.31 & 1.33 & 1.30  & 1.33 & 1.29 & 1.32\\
    &VCD-FL & 24895 & 25888 & 27119 & 27831 & 28978 & 5385 & 10687  & 16011 & 21328 & 28961 \\
    \midrule
    \multirow{2}{*}{Decryption (ms)} &EIFL &0.54 & 1.13 & 1.73 & 2.34 & 2.88 & 2.86 & 2.83  & 2.76 & 2.86 & 2.88\\
    &VCD-FL &2151  &2127 & 2172 & 2124 & 2122 & 424 & 853 & 1281 & 1707 & 2122 \\
    \midrule
    \multirow{2}{*}{Total Time (ms)}
    & EIFL & 0.80 & 1.65 & 2.51 & 3.39 & 4.19 & 4.19 & 4.13 & 4.09 & 4.15 & 4.20 \\
    & VCD-FL & 27046 & 28015 & 29291 & 29955 & 31100 & 5809 & 11540 & 17292 & 23035 & 31083 \\
    \bottomrule
  \end{tabular}
  \label{tbl:diff}
\end{table*}

\begin{figure*}
     \centering
     \begin{subfigure}[b]{0.24\textwidth}
         \centering
         \includegraphics[scale=0.28]{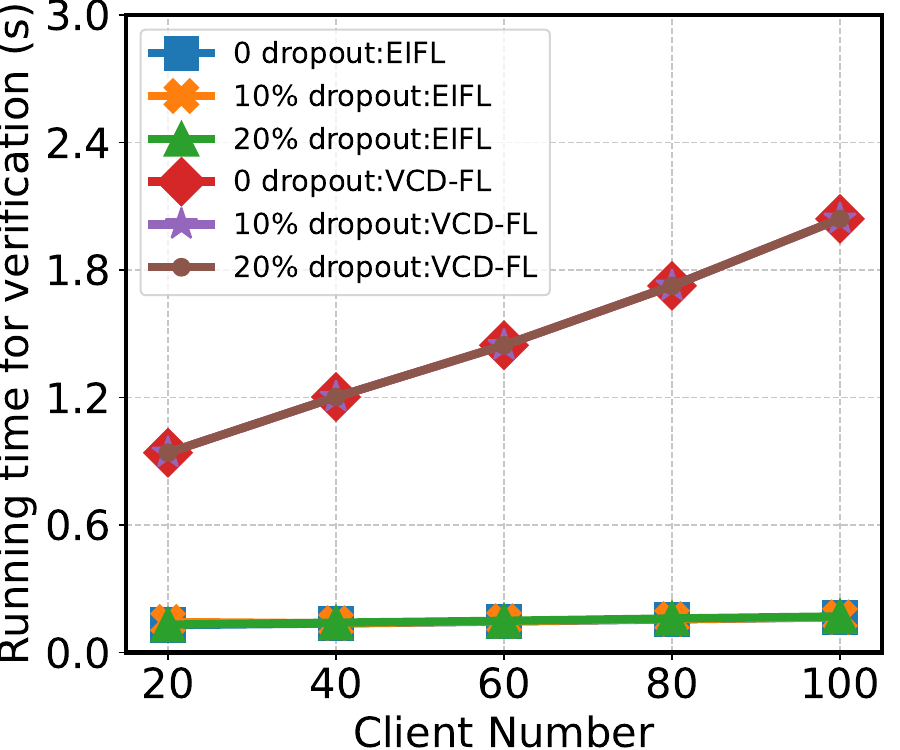}
         \caption{$M=100K$}
         \label{fig:client-runtime:a}
     \end{subfigure}
     \begin{subfigure}[b]{0.24\textwidth}
         \centering
         \includegraphics[scale=0.28]{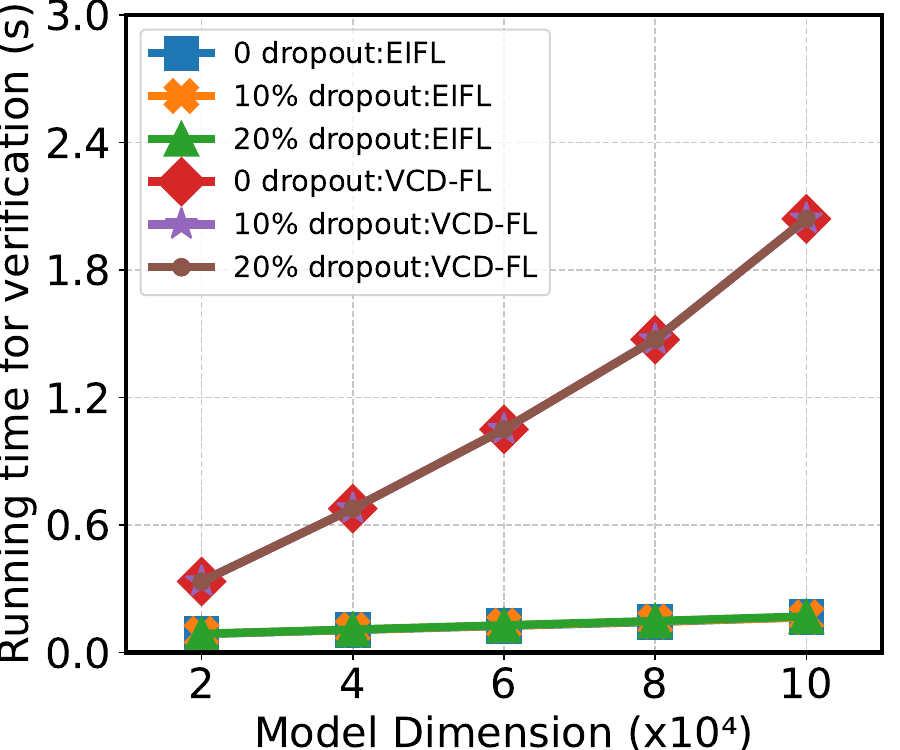}
         \caption{$N=100$}
         \label{fig:client-runtime:b}
     \end{subfigure}
     \begin{subfigure}[b]{0.24\textwidth}
         \centering
         \includegraphics[scale=0.28]{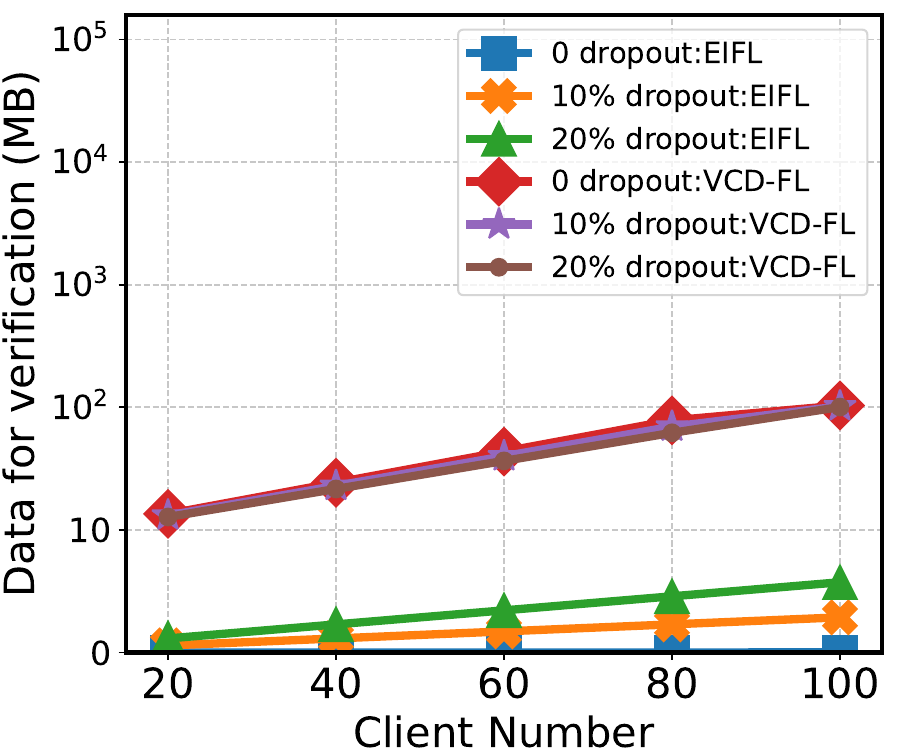}
         \caption{$M=100K$}
         \label{fig:client-data:c}
     \end{subfigure}
     \begin{subfigure}[b]{0.24\textwidth}
         \centering
         \includegraphics[scale=0.28]{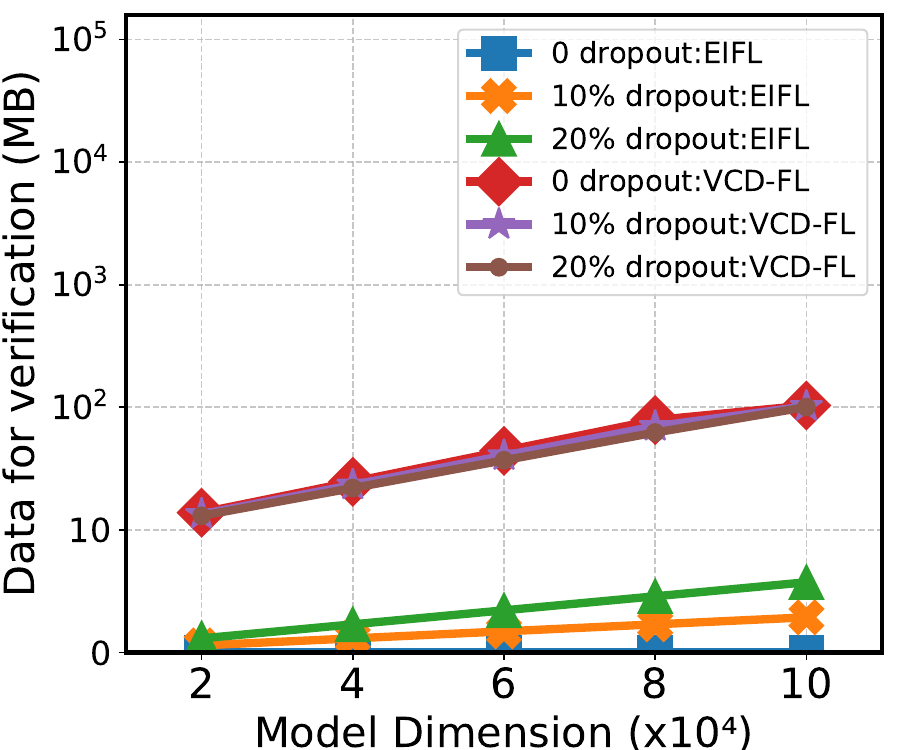}
         \caption{$N=100$}
         \label{fig:client-data:d}
     \end{subfigure}
     \vspace{-1mm}
        \caption{Comparison of verification performance with different client verification dropout rates. (a) Running time with different numbers of clients $N$. (b) Running time with different model dimension $M$. (c) Data transfer with different numbers of clients $N$. (d) Data transfer with different model dimension $M$.} 
        \label{fig:verify}
        \vspace{-4mm}
\end{figure*}

\subsubsection{Data transfer volume} 
Figures \ref{fig:verification-data:c} and \ref{fig:verification-data:d} illustrate the comparison between the verification data transfer volume and the total data transfer volume. Combined with the results summarized in Table \ref{tbl:running time}, we observe that the total data transfer volume increases with the increase in the number of clients or the model dimension. As the number of clients increases, the additional \textcolor{blue}{overhead} introduced by output integrity verification grows slowly and is negligible compared with the total communication of the protocol. More notably, the data for verification remains constant as the model dimension increases. Recall that the communication related to verification includes: (1) Sending the masked inner product value $Com_n$ and its signature $\sigma_n^c$ to the server, (2) Receiving other clients' masked inner product values and signatures from the server. Since the length of the masked inner product value and signature is fixed, the communication complexity for verification is $O(N)$. This theoretical result directly explains the extremely lightweight nature of the verification communication overhead of EIFL.

\subsection{Comparison with Previous Work}
\subsubsection{Verification performance for output integrity} 
We compare EIFL with existing SOTA schemes including VERIFL \cite{DBLP:cite8}, VERSA \cite{vsersa}, VOSA \cite{VOSA} and VCD-FL \cite{VCD-FL}. As illustrated in Figures \ref{fig:diff_ver_time:a} and \ref{fig:diff_ver_time:b}, as the number of clients and the model dimension increase, our EIFL consistently maintains the shortest verification time. VERIFL has the longest verification time because the linearly homomorphic hash it uses is realized over the NIST P-256 elliptic curve. Figures \ref{fig:diff_ver_com:c} and \ref{fig:diff_ver_com:d} present the communication overhead for verification. We can observe that EIFL has the smallest communication overhead. Since both EIFL and VERIFL convert input vectors into a single evidential value, their communication overhead is independent of the model dimension and only increases as the number of clients increases. By contrast, VCD-FL requires each client to convert its input vector into a commitment vector of the same dimension and exchange them with each other, resulting in significant communication overhead. Similarly, VERSA and VOSA require each client to generate an evidential vector with the same dimension as its original input vector and upload it to the server, which means that their communication overhead is primarily influenced by the model dimension.

\subsubsection{Addressing the verification vulnerability and dropout clients} To address the verification vulnerability, EIFL and VCD-FL respectively propose auxiliary information negotiation methods to assist verification. Since clients in EIFL agree on the auxiliary information and upload evidence (i.e., masked inner product values) only after the server completes aggregation, this implies that clients who have uploaded input vectors but dropped out before uploading evidence cannot participate in the verification. EIFL addresses this issue through a resending operation, but incurs additional overheads. We evaluate the impact of dropout rate and compare EIFL with VCD-FL. As shown in Figures \ref{fig:client-data:c} and \ref{fig:client-data:d}, the dropout rate increases the communication overhead of EIFL while decreasing that of VCD-FL. This is because the server in EIFL needs to resend the masked input vectors of clients that dropped out during the verification phase to the remaining clients. Notably, clients in EIFL agree on a random vector via our proposed generation method and do not require additional communication. For VCD-FL, although the dropout clients reduce the number of peer-to-peer communications between remaining clients when agreeing on a random vector, its communication overhead remains substantial. Overall, the communication overhead of EIFL remains much lower than that of VCD-FL. From these results, it can be concluded that our proposed random vector generation method significantly reduces the communication overhead for clients to agree on a random vector. In addition, as shown in Figures \ref{fig:client-runtime:a} and \ref{fig:client-runtime:b}, EIFL outperforms VCD-FL in terms of verification time, and the impact of the dropout rate on the verification time of both schemes is negligible. 

\subsubsection{Performance for output privacy}
Since both EIFL and VCD-FL support output privacy protection, we also compare EIFL with VCD-FL in terms of encryption and decryption time. Table \ref{tbl:diff} presents the additional overheads for output privacy of EIFL and VCD-FL. We can observe that both the encryption time and the decryption time of EIFL are much shorter than those of VCD-FL. This is mainly because during the encryption phase, VCD-FL uses Lagrange interpolation to convert an input vector into polynomial coefficients, and during the decryption phase, it computes the polynomial result to retrieve the output result. For EIFL, protecting output privacy only requires clients to encrypt the information uploaded to the server using AES in Round 4, and decrypt it in Round 5. 

\section{Conclusion}
\label{conclusion}
In this paper, we focus on the output privacy and verification vulnerability of output integrity  and address the above challenges by proposing a novel \emph{provably} privacy-preserving FL method, called EIFL. Specifically, EIFL adopts a two-stage aggregation and combines it with symmetric encryption to protect the output privacy. To address the verification vulnerability of output integrity, EIFL introduces an efficient verification method based on vector inner product and a random vector generation method. Compared with other methods, EIFL innovatively binds the auxiliary information to the output integrity, and eliminates the requirement for keeping the auxiliary information confidential from the server. 
\bibliographystyle{IEEEtran}
\bibliography{IEEEabrv,main}

\end{document}